\documentclass[11pt]{article}

\usepackage{fullpage}
\usepackage{hdb_macros}
\usepackage{tikz}
\usepackage{pgfplots}
\usepackage{pgfplotstable}
\usetikzlibrary{positioning}

\usepackage[draft,multiuser,inline,nomargin]{fixme}
\fxusetheme{color}
\FXRegisterAuthor{h}{eh}{\color{blue}Huck}
\FXRegisterAuthor{k}{ek}{\color{purple}Karthik}
\FXRegisterAuthor{p}{ep}{\color{teal}Philip}
\FXRegisterAuthor{e}{ee}{\color{teal}Evelyn}
\FXRegisterAuthor{s}{es}{\color{orange}Sasha}

\renewcommand{\C}{\mathcal{C}}
\newcommand{\Cperp}{\mathcal{C}^{\perp}}
\newcommand{\dualMMexp}{\ensuremath{\omega^{\perp}}}
\newcommand{\cR}{{\mathcal{R}}}

\newcommand{\QP}{\cc{QP}}

\newcommand{\full}[1]{#1} %

\title{Matrix Multiplication Verification Using Coding Theory}
\author{
    Huck Bennett\thanks{University of Colorado Boulder. Email: \email{huckbennett@gmail.com}. Supported by NSF Grant CCF-2312297. Most of this work was done while the author was at Oregon State University.}
    \and
    Karthik Gajulapalli\thanks{Georgetown University. Email: \email{kg816@georgetown.edu}. Supported by NSF grant CCF-2338730.}
     \and
 	Alexander Golovnev\thanks{Georgetown University. Email: \email{alexgolovnev@gmail.com}. Supported by NSF grant CCF-2338730.}
 	\and
        Evelyn Warton\thanks{Oregon State University. Email: \email{wartonp@oregonstate.edu}.}
 }
\date{}
\pgfplotsset{compat=1.18} 
\begin{document}
\pagenumbering{roman}
\maketitle
\listoffixmes

\begin{abstract}
We study the \emph{Matrix Multiplication Verification Problem} (MMV) where the goal is, given three $n \times n$ matrices $A$, $B$, and $C$ as input, to decide whether $AB = C$.
A classic randomized algorithm by Freivalds (MFCS, 1979) solves MMV in $\Ot(n^2)$ time, and a longstanding challenge is to (partially) derandomize it while still running in faster than matrix multiplication time (i.e., in $o(n^{\omega})$ time).

To that end, we give two algorithms for MMV in the case where $AB - C$ is \emph{sparse}. Specifically, when $AB - C$ has at most $O(n^{\delta})$ non-zero entries for a constant $0 \leq \delta < 2$, we give (1) a deterministic $O(n^{\omega - \eps})$-time algorithm for constant $\eps = \eps(\delta) > 0$, and (2) a randomized $\Ot(n^2)$-time algorithm using $\delta/2 \cdot \log_2 n + O(1)$ random bits. 
The former algorithm is faster than the deterministic algorithm of K\"{u}nnemann (ESA, 2018) when $\delta \geq 1.056$, and the latter algorithm uses fewer random bits than the algorithm of Kimbrel and Sinha (IPL, 1993), which runs in the same time and uses $\log_2 n + O(1)$ random bits (in turn fewer than Freivalds's algorithm).

Our algorithms are simple and use techniques from coding theory. Let $H$ be a parity-check matrix of a Maximum Distance Separable (MDS) code, and let $G = (I \stbar G')$ be a generator matrix of a (possibly different) MDS code in systematic form.
Our deterministic algorithm uses fast rectangular matrix multiplication to check whether $HAB = HC$ and $H(AB)^T = H(C^T)$, and our randomized algorithm samples a uniformly random row $\vec{g}'$ from $G'$ and checks whether $\vec{g}' AB  = \vec{g}' C$ and $\vec{g}' (AB)^T =  \vec{g}' C^T$.

We additionally study the \emph{complexity} of MMV.
We first show that all algorithms in a natural class of deterministic linear algebraic algorithms for MMV (including ours) require $\Omega(n^{\omega})$ time.
We also show a barrier to proving a super-quadratic running time lower bound for matrix multiplication (and hence MMV) under the Strong Exponential Time Hypothesis~(SETH). 
Finally, we study relationships between natural variants and special cases of MMV (with respect to deterministic $\Ot(n^2)$-time reductions).
\end{abstract}

\thispagestyle{empty}
\newpage

\tableofcontents
\newpage
\pagenumbering{arabic}

\section{Introduction}
\label{sec:introduction}

The goal of the \emph{Matrix Multiplication Problem} (MM) is to compute the product $AB$ of two $n \times n$ matrices $A$ and $B$ given as input. Matrix multiplication has many practical and theoretical applications, and because of this has been studied by an extensive line of work. The primary goal of this work has been to determine the running time $O(n^{\omega})$ of the fastest algorithms for MM, which is captured by the matrix multiplication exponent $\omega$.%
\footnote{Formally, $\omega$ is defined as the infimum over $\omega'$ such that the product of two $n \times n$ matrices can be computed in $O(n^{\omega'})$ time. So, MM algorithms are actually only guaranteed to run in $O(n^{\omega + \eps})$ time for any constant $\eps > 0$.}
The best upper bounds on $\omega$ and related quantities continue to improve~\cite{gall2018improved,alman2021refined,duan2022faster,gall2023faster,williams2024new}, and~\cite{williams2024new} recently showed the current best known bound of $\omega \leq 2.371552$.
The dream of this line of work is to show that $\omega = 2$, and this in fact holds under certain plausible combinatorial and group-theoretic conjectures (see~\cite[Conjecture~4.7 and Conjecture~3.4]{cohn2005group}).
Nevertheless, showing that $\omega = 2$ seems very challenging for the time being.

In this work, we consider a variant of matrix multiplication where the goal is to \emph{verify} that the product of two matrices is equal to a third matrix. Specifically, we study the \emph{Matrix Multiplication Verification Problem} (MMV) where, given three $n \times n$ matrices $A$, $B$, and $C$ as input, the goal is to decide whether $AB = C$.  MMV is clearly no harder than matrix multiplication---it can be solved in $O(n^{\omega})$ time by computing the product $AB$ and then comparing the product entry-wise against $C$---but it is natural to ask whether it is possible to do better. In what became classic work, Freivalds~\cite{freivalds1979fast} answered this question in the affirmative and gave a simple, randomized algorithm that solves MMV in $\Ot(n^2)$ time. 
This $\Ot(n^2)$ running time bound is essentially the best possible, and so, unlike matrix multiplication, the complexity of MMV is relatively well understood. 

However, it is in turn natural to ask whether it is possible to \emph{derandomize} Freivalds's algorithm partially or completely. More specifically, it is natural to ask whether it is possible to give a  \emph{deterministic} algorithm for MMV running in $\Ot(n^2)$ time or at least $O(n^{\omega - \eps})$ time for constant $\eps > 0$.%

\footnote{We use the notation $\Ot(f(n))$ to mean $f(n) \cdot \poly(\log f(n))$. Freivalds's algorithm uses $O(n^2)$ arithmetic operations, each of which takes $\poly(\log n)$ time when working over integer matrices with entries bounded in magnitude by $\poly(n)$; we assume this setting in the introduction. \\
\indent Of course, it is only possible for such a $O(n^{\omega - \eps})$-time algorithm to exist if $\omega > 2$. We assume that this is the case throughout the introduction.}
Or, if it is not possible to give a deterministic algorithm for MMV with these running times, it is natural to ask whether it is possible to use fewer random bits than Freivalds's algorithm, which uses $n$ random bits.
Trying to answer these questions has become a key goal for derandomization efforts, and has received substantial study~\cite{alon-1990,kimbrel1993probabilistic,korec2014deterministic,kunnemann2018nondeterministic}.

\begin{table}
    \centering
    \bgroup
    \def\arraystretch{1.5}%
    \begin{tabular}{|c|c|c|c|c|}
        \hline
         \textbf{Algorithm} & \textbf{Asymptotic Runtime} & \textbf{Bits of Randomness}\\
         \hline
         Matrix Multiplication & $n^{\omega + \eps}$ & 0 \\
         \hline
         Random Entry Sampling (folklore) & $n^{3-\delta}$ & $2n^{2-\delta}  \cdot \log_2(n) + O(1)$ \\
         \hline
         Freivalds's Algorithm \cite{freivalds1979fast} & $n^2$ & $n$ \\
         \hline
        Vandermonde Mat. Sampling \cite{kimbrel1993probabilistic} & $n^2$ & $\log_2(n) + O(1)$ \\
         \hline
         Multipoint Poly. Evaluation \cite{kunnemann2018nondeterministic} & $n^2+n^{1 + \delta}$ & 0 \\
        \hline
        Cauchy Bound \cite{korec2014deterministic} & $n^3$ ($n^2$ in Integer RAM) & 0 \\
         \hline
            \textbf{Parity Check/Fast RMM (Thm.~\ref{thm:rmm-deterministic-intro})} & $n^{\omega(1, 1, \delta/2) + \eps}$ & 0 \\
         \hline 
            \textbf{Cauchy Mat. Sampling (Thm.~\ref{thm:fast-random-intro})} & $n^2$ & $\frac{\delta}{2} \cdot \log_2(n) + O(1)$ \\
         \hline
    \end{tabular}
    \egroup
    \caption{
    Algorithms for MMV on matrices $A,B,C \in \mathbb{Z}^{n \times n}$ with entries of magnitude at most $\poly(n)$ and such that $AB - C$ has at most $n^\delta$ non-zero entries for $0 \leq \delta \leq 2$. Our new algorithms are shown in bold.
        We list asymptotic running times, with $\poly(\log n)$ factors suppressed for readability, and the number of random bits used to achieve success probability $1/2$. (Each of the three listed randomized algorithms has one-sided error, so this probability is meaningful.) Here $\omega(\cdot, \cdot, \cdot)$ is the rectangular matrix multiplication exponent, $\omega = \omega(1, 1, 1)$ is the (square) matrix multiplication exponent, and $\eps > 0$ is an arbitrarily small positive constant.}
   \label{tbl:mmv-algs}
\end{table}

\pgfplotsset{compat=1.18, width = 10cm, height=10cm}

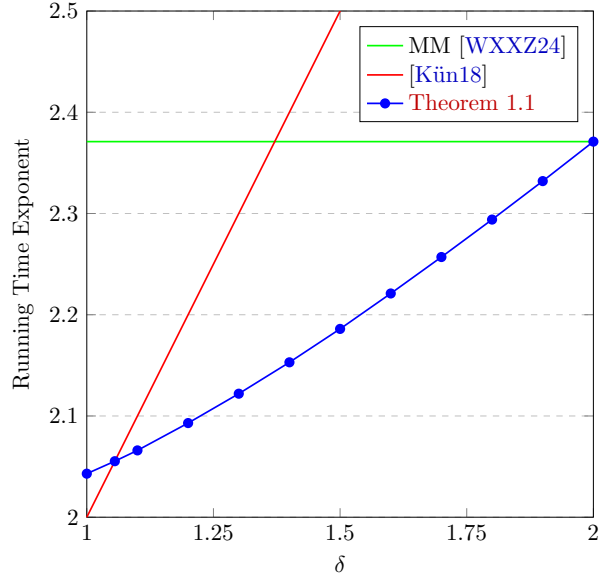
\begin{figure}[t]
\begin{center}
\begin{tikzpicture}[scale=0.8]
\begin{axis}[
    xlabel={$\delta$},
    ylabel={Running Time Exponent},
    xmin=1.00, xmax=2.00,
    ymin=2.00, ymax=2.50,
    xtick={1.00, 1.25, 1.50, 1.75, 2.00},
    ytick={2.0, 2.1, 2.2, 2.3, 2.4, 2.5},
    legend pos=north east,
    ymajorgrids=true,
    grid style=dashed,
    legend cell align=left,
    legend pos=north east,
    scatter/classes={%
    a={mark=*,fill=blue}}
]

\addplot[
    color=green,
    mark=circle,
    style=thick
    ]
    coordinates {
    (1, 2.371)(1.055322, 2.371)(1.1, 2.371)(1.2, 2.371)(1.3, 2.371)(1.4, 2.371)(1.5, 2.371)(1.6, 2.371)(1.7, 2.371)(1.8, 2.371)(1.9, 2.371)(2.0,2.371)
    };
    \addlegendentry{MM \cite{williams2024new}}

\addplot[
    color=red,
    mark=circle,
    style=thick
    ]
    coordinates {
    (1, 2.0)(1.055322, 2.055322)(1.1, 2.1)(1.2, 2.2)(1.3, 2.3)(1.4, 2.4)(1.5, 2.5)(1.6, 2.6)(1.7, 2.7)(1.8, 2.8)(1.9, 2.9)(2.0,3.0)
    };
    \addlegendentry{\cite{kunnemann2018nondeterministic}}

\addplot[scatter,
    scatter src=explicit symbolic,
    color=blue,
    mark=circle,
    style=thick
    ]
    table[meta=label] {
    x y label
1 2.043 a
1.055322 2.055322 a
1.1 2.066 a 
1.2 2.093 a
1.3 2.122 a
1.4 2.153 a
1.5 2.186 a
1.6 2.221 a
1.7 2.257 a
1.8 2.294 a
1.9 2.332 a
2.0 2.371 a
    };
    
    \addlegendentry{\cref{thm:rmm-deterministic-intro}}

\end{axis}
\end{tikzpicture}
\end{center}

\caption{
Running times of deterministic algorithms for MMV when $AB - C$ is $O(n^{\delta})$-sparse for $1 \leq \delta \leq 2$.
Our algorithm from \cref{thm:rmm-deterministic-intro} is faster than the best known algorithms for matrix multiplication~\cite{williams2024new} and faster than K\"{u}nnemann's algorithm~\cite{kunnemann2018nondeterministic} for all $1.056 \leq \delta < 2$. The plotted blue points corresponding to the running time of the algorithm in \cref{thm:rmm-deterministic-intro} are derived from the bounds on $\omega(1, 1, \delta/2)$ in~\cite[Table 1]{williams2024new}. The line segments connecting them are justified by the fact that $\omega(1, 1, \cdot)$ is a convex function.}

\label{fig:alg-running-times}
\end{figure}

\subsection{Our Results}
\label{sec:our-results}

Our main results are two new algorithms for the Matrix Multiplication Verification Problem in the \emph{sparse} regime, i.e., in the case where $AB - C$ is promised to have few non-zero entries (if any). 
See \cref{tbl:mmv-algs} for a summary of our algorithms and how they compare to other known algorithms for MMV.
Additionally, we give a barrier for giving a fast algorithm for MMV using a broad class of linear algebraic techniques, a barrier to showing hardness of MMV, and reductions between variants of MMV.

\subsubsection{Algorithms}
\label{sec:summary-algorithms}

Besides being inherently interesting, MMV in the sparse regime is the natural decision version of the well-studied \emph{Output-Sensitive Matrix Multiplication Problem} (OSMM). It is also motivated by the following scenario. Suppose that Alice wants to compute the product $AB$ of two large matrices $A$ and $B$, but has restricted computational resources. So, she sends $A$ and $B$ to Bob, who has more extensive computational resources. Bob computes the product $AB$, and sends the result back to Alice over a noisy channel (without error-correction, which increases the size of the message), from which Alice receives a matrix $C$. Alice knows that either $C = AB$ as desired, or that $C$ is corrupted but (with high probability) only differs from $AB$ in a few entries. She wants to check which is the case efficiently.

We define $\norm{\vec{v}}_0$ (respectively, $\norm{M}_0$) to be the number of non-zero entries in (i.e., Hamming weight of) a vector $\vec{v}$ (respectively, matrix $M$). We call a vector $\vec{v}$ (respectively, matrix $M$) \emph{$t$-sparse} if $\norm{\vec{v}}_0 \leq t$ (respectively, if $\norm{M}_0 \leq t$).

Our first algorithm is deterministic, and
uses fast rectangular matrix multiplication. For $\alpha, \beta, \gamma \in [0, 1]$, let the rectangular matrix multiplication exponent $\omega(\alpha, \beta, \gamma)$ be the infimum over values $\omega' > 0$ such that the product of a $n^{\alpha} \times n^{\beta}$ matrix and a $n^{\beta} \times n^{\gamma}$ matrix can be computed using $O(n^{\omega'})$ arithmetic operations. Note that $\omega = \omega(1, 1, 1)$ is the standard (square) matrix multiplication exponent.

\begin{theorem}[Fast deterministic MMV for sparse matrices, informal]
\label{thm:rmm-deterministic-intro}
Let $A, B, C \in \Z^{n \times n}$ be matrices satisfying $\max_{i,j} \set{\abs{A_{i,j}}, \abs{B_{i,j}}, \abs{C_{i,j}}} \leq n^c$ for some constant $c > 0$ and 
satisfying ${\norm{AB - C}_0 \leq n^{\delta}}$ for $0 \leq \delta \leq 2$.
Then for any constant $\eps > 0$, there is a deterministic algorithm for MMV on input $A, B, C$ that runs in $O(n^{\omega(1, 1, \delta/2) + \eps})$ time.
\end{theorem}

We note that $\omega(1, 1, \beta) < \omega$ for all $\beta < 1$ (assuming $\omega > 2$; see \cref{thm:rectMM-ub}), and so our algorithm is faster than matrix multiplication when $AB - C$ is promised to be $O(n^{\delta})$-sparse for constant $\delta < 2$. Furthermore, it is faster than K\"{u}nnemann's algorithm~\cite{kunnemann2018nondeterministic}, which is also for MMV in the regime where $AB - C$ is sparse, when $\omega(1, 1, \delta/2) < 1 + \delta$. 
The equation $\omega(1, 1, \delta/2) = 1 + \delta$ whose unique solution corresponds to the crossover point at which our algorithm becomes faster than K\"{u}nnemann's turns out to be relevant in other contexts too~\cite{journals/jacm/Zwick02}, and~\cite{gall2018improved,williams2024new} both provide bounds on its solution. Specifically,~\cite{williams2024new} shows that the solution $\delta$ to this equation satisfies $\delta \leq 1.056$, and so our algorithm in \cref{thm:rmm-deterministic-intro} is (strictly) faster than any previously known deterministic algorithm for MMV when $1.056 \leq \delta < 2$. See \cref{fig:alg-running-times}.

Additional bounds on $\omega(1, \delta/2, 1) = \omega(1, 1, \delta/2)$---and hence the running time of the algorithm in \cref{thm:rmm-deterministic-intro}---appear in~\cite[Table 1]{williams2024new}. For example, that table shows that $\omega(1, 1, 0.55) < 2.067$ and $\omega(1, 1, 0.95) < 2.333$ (which correspond to $\delta = 1.1$ and $\delta = 1.9$, respectively). We again refer the reader to \cref{fig:alg-running-times}.
We also note that our algorithm runs in essentially optimal $\Ot(n^2)$ time when $\delta \leq 0.642 \leq 2 \omega^{\perp}$, where $\omega^{\perp} := \sup \set{\omega' > 0 : \omega(1, 1, \omega') = 2} \geq 0.321$ is the dual matrix multiplication exponent~\cite{williams2024new}, but that K\"{u}nnemann's algorithm~\cite{kunnemann2018nondeterministic} runs in $\Ot(n^2)$ time for any $\delta \leq 1$.

Our second algorithm runs in $\Ot(n^2)$ time, but is randomized. It uses few bits of randomness when $AB - C$ is sparse.

\begin{theorem}[Fast randomized MMV for sparse matrices, informal] \label{thm:fast-random-intro}
Let $c > 0$ be a constant, let $A, B, C \in \Z^{n \times n}$ be matrices 
satisfying $\max_{i,j} \set{\abs{A_{i,j}}, \abs{B_{i,j}}, \abs{C_{i,j}}} \leq n^c$ and 
satisfying ${\norm{AB - C}_0 \leq n^{\delta}}$ for $0 \leq \delta \leq 2$, and let $\eps = \eps(n) \geq 1/n$.
Then there is a randomized algorithm for MMV on input $A, B, C$ that runs in $\Ot(n^2)$ time, succeeds with probability $1 - \eps$, and uses at most $\lceil \delta/2 \cdot \log_2(n) + \log_2(1/\eps) \rceil$ bits of randomness.
\end{theorem}

\cref{thm:fast-random-intro} improves on the number of random bits used by the algorithm of Kimbrel and Sinha~\cite{kimbrel1993probabilistic} when $\delta < 2$ (which uses $\log_2(n) + \log_2(1/\eps) + O(1)$ random bits regardless of the sparsity of $AB - C$), and matches the number of random bits used by their algorithm when $\delta = 2$. The algorithms both run in $\Ot(n^2)$ time.
In fact, one may think of the algorithm summarized in \cref{thm:fast-random-intro} as a natural extension of the algorithm in~\cite{kimbrel1993probabilistic} to handle the sparse case more efficiently, although it requires additional techniques to implement. (Our algorithm requires matrices with a stronger pseudorandomness property than theirs; see the ``algorithmic techniques'' section below.)

We note that \cref{thm:fast-random-intro} only improves on known algorithms when $1 < \delta < 2$, and only by a factor of $\delta/2$. Indeed, as mentioned above, when $\delta \leq 1$ K\"{u}nnemann's algorithm~\cite{kunnemann2018nondeterministic} solves MMV \emph{deterministically} in $\Ot(n^2)$ time, and when $\delta = 2$ our algorithm matches the number of random bits used by Kimbrel and Sinha's algorithm. 
Although seemingly modest, this constant-factor improvement is not surprising: any super-constant improvement on the number of bits used by~\cite{kimbrel1993probabilistic} (i.e., an MMV algorithm using $o(\log n)$ random bits) could be turned into a deterministic algorithm for MMV with only a sub-polynomial (i.e., $n^{o(1)}$) multiplicative increase in running time.

\paragraph{Algorithmic techniques.}
Here we briefly summarize the techniques that we use for the MMV algorithms corresponding to \cref{thm:rmm-deterministic-intro,thm:fast-random-intro}.
We start by remarking that \cref{thm:rmm-deterministic-intro,thm:fast-random-intro} hold not just for matrices over $\Z$ with entries of polynomial magnitude, but also for matrices over all finite fields $\F_q$ with $q \leq \poly(n)$.%
\footnote{The algorithms also work over larger finite fields, but with slower running times due to the increased bit complexity of performing arithmetic operations over those fields.}
In fact, our algorithms work ``natively'' in the finite field setting---i.e., on $n \times n$ matrices $A, B, C$ over finite fields $\F_q$---which is directly amenable to using techniques from coding theory. We assume this setting in the description below.
Furthermore, there is a linear-time, sparsity-preserving reduction from $\MMV$ to the special case of $\MMV$ where $C$ is fixed as $C = 0$ and the goal is to decide whether $AB = 0$ for input matrices $A, B$; see \cref{def:all-zeroes,thm:mmv-all-zeroes-reduction}. We will also generally assume this setting in the introduction.

For both of our algorithms, we will use the observation that if $AB$ is non-zero and $t$-sparse then at least one row or column of $AB$ must be non-zero and $k$-sparse for $k := \lfloor \sqrt{t} \rfloor$. A similar observation appears in~\cite{wu2023correcting}. %

Our first, deterministic algorithm (\cref{thm:rmm-deterministic-intro}) uses a matrix $H$ over $\F_q$ such that any $k$ columns of $H$ are linearly independent. 
Equivalently, we require a matrix $H \in \F_q^{m \times n}$ such that for all non-zero vectors $\vec{x} \in \F_q^n$ with $\norm{\vec{x}}_0 \leq k$ (corresponding to a sparse, non-zero column or row of $AB$), $H\vec{x} \neq \vec{0}$. 
This is exactly the property guaranteed by the parity-check matrix $H$ of an error correcting code $\mathcal{C} \subseteq \F_q^n$ with minimum distance $d > k$. Moreover, if a code $\C$ with minimum distance $d = k + 1$ is a so-called Maximum Distance Separable (MDS) code, then it has a $k \times n$ parity-check matrix $H$.
MDS codes with useful parameters exist and have efficiently constructible parity-check matrices. In particular, (generalized) Reed-Solomon codes are MDS codes, and exist when $k \leq n \leq q$ (see, e.g.,~\cite{hall-coding}). Their parity-check matrices $H$ are Vandermonde matrices, which are constructible in $k n \cdot \poly(\log q) \leq n^2 \cdot \poly(\log q)$ time.

Our algorithm then uses fast rectangular matrix multiplication to compute $HAB = (HA)B$ and $H(AB)^T = (HB^T) A^T$ using roughly $n^{\omega(1, 1, \delta/2)}$ arithmetic operations, where $0 \leq \delta \leq 2$ is such that $t \leq n^{\delta}$. If $AB = 0$, then $HAB = H(AB)^T = 0$. On the other hand, if $AB \neq 0$ then $AB$ is $t$-sparse and therefore has a $k$-sparse row or column. So, at least one of the expressions $HAB$ and $H(AB)^T$ is non-zero.

Our second, randomized algorithm (\cref{thm:fast-random-intro}) uses a matrix $S \in \F_q^{m \times n}$ with the property that all of its $k \times k$ submatrices are non-singular.
Matrices $S$ with this property are called \emph{$k$-regular}, and matrices $S$ all of whose square submatrices (of any size) are non-singular are called \emph{super regular} (see also the definitions in \cref{sec:regular-matrices}).
We note that $k$-regularity is stronger than the property we require for $H$ in the first algorithm. In particular, if a matrix $S \in \F_q^{m \times n}$ is $k$-regular and $0 < \norm{\vec{x}}_0 \leq k$, then $\norm{S\vec{x}}_0 \geq m - k + 1$. I.e., $S$ being $k$-regular implies not only that $S\vec{x}$ is non-zero, but that $S\vec{x}$ has relatively high Hamming weight for such $\vec{x}$.
This property is useful because it implies that $\Pr[\iprod{\vec{s}, \vec{x}} \neq 0] \geq (m - k + 1)/m$, where $\vec{s}$ is a random row of $S$. 
Indeed, this observation leads to our second algorithm: we sample a random row $\vec{s}$ from a $k$-regular matrix $S \in \F_q^{m \times n}$ and check whether $\vec{s}AB = \vec{0}$ and $\vec{s}(AB)^T = \vec{0}$. 
Setting, e.g., $m = 2k$, we get that this algorithm succeeds with probability at least $(2k - k + 1)/(2k) > 1/2$.

It remains to construct (rows of) $k$-regular matrices $S$ efficiently.
Although a priori it is not even obvious that $k$-regular matrices exist for arbitrary $k$, in fact \emph{super regular} matrices exist and are efficiently constructible. Specifically, we use a family of super regular (and hence $k$-regular) matrices called \emph{Cauchy matrices}; the entries of such a matrix $S$ are defined as $S_{i,j} = 1/(x_i - y_j)$, where $x_1, \ldots, x_m, y_1, \ldots, y_n$ are distinct elements of $\F_q$.
In fact, as follows from their definition, given a (random) index $1 \leq i \leq m$, it is even possible to construct the $i$th row of a Cauchy matrix $S$ efficiently without computing the rest of the matrix, as needed.%

Finally, we remark that there is a deep connection between MDS codes and super regular matrices (and between generalized Reed-Solomon codes, Vandermonde matrices, and Cauchy matrices). Specifically, if $G = (I \stbar S)$ is the generator matrix of an MDS code in systematic form, then $S$ is a super regular matrix~\cite{roth1989mds}.
Moreover, if such a matrix $G$ is the generator matrix of a generalized Reed-Solomon code, then $S$ is a Cauchy matrix~\cite{roth1989mds}. See also \cref{rem:mds-vandermonde-cauchy}.

\subsubsection{Barriers}
\label{sec:summary-barriers}

The dream for the line of work described in this paper is to give a deterministic, $\Ot(n^2)$-time algorithm for MMV on arbitrary matrices. 
However, achieving this goal has proven to be very difficult despite a substantial amount of work towards it. So, it is natural to ask whether perhaps no such algorithm exists, i.e., whether MMV is in some sense \emph{hard}. We first show a result in this direction, and then show a barrier result to showing SETH hardness of MMV (and even MM).%
\footnote{More properly, our first result is a barrier to giving a fast algorithm for MMV, and our second result is a barrier to showing hardness of MMV (i.e., it ``gives a barrier to giving a barrier'' for a fast MMV algorithm).}

\paragraph{Linear algebraic algorithms barrier.} 
We first prove that a natural class of deterministic linear algebraic algorithms for MMV based on multiplying smaller matrices---including the algorithm in \cref{thm:rmm-deterministic-intro}---cannot run in less than $n^{\omega}$ time using when instantiated with a matrix multiplication subroutine running in worst-case rectangular matrix multiplication time and when performing all multiplications independently. Specifically, the $\Omega(n^{\omega})$ lower bound holds if for all $\alpha, \beta \geq 0$, the subroutine requires $\Omega(n^{\omega(1, 1, \alpha)})$ to compute the product of an $n \times n^{\alpha}$ matrix and an $n \times n$, and $\Omega(n^{\omega(1, 1, \beta)})$ time to compute the product of an $n \times n$ matrix and an $n \times n^{\beta}$ matrix.

The idea is that natural algorithms for verifying that $AB = C$ for $n \times n$ matrices $A, B, C$ including ours amount to performing $k$ ``zero tests.''
\footnote{Here we describe our barrier for MMV, but in \cref{sec:linealgbar} we present it for the closely related All Zeroes Problem.}
More specifically, the $i$th such test checks that $L_i (AB - C) R_i = 0$ for some fixed $n^{\alpha_i} \times n$ matrix $L_i$ and $n \times n^{\beta_i}$ matrix $R_i$, where $\alpha_i, \beta_i \in [0, 1]$. We observe that the conditions $L_i (AB - C) R_i = 0$ for $i = 1, \ldots, k$ together correspond to a homogeneous system of $\sum_{i=1}^k n^{\alpha_i + \beta_i}$ linear equations in the $n^2$ variables corresponding to the entries of $X = AB - C$ for $1 \leq i, j \leq n$. So, for this system to have $X_{i,j} = 0$ for $1 \leq i, j \leq n$ as its unique solution, it must be the case that $\sum_{i=1}^k n^{\alpha_i + \beta_i} \geq n^2$, which we show implies that $\sum_{i=1}^k n^{\omega(1, 1, \min(\alpha_i, \beta_i))} \geq \sum_{i=1}^k n^{\omega(\alpha_i, 1, \beta_i)} \geq n^{\omega}$.
Therefore, an algorithm that independently computes each product $L_i AB R_i$ in time $\Omega(n^{\omega(1, 1, \min(\alpha_i, \beta_i))})$ uses $\Omega(n^{\omega})$ time.
This barrier appears in \cref{thm:lin-alg-barrier,cor:lin-alg-algo-barrier}.

\paragraph{A barrier to SETH-hardness of MM.}
While under certain reasonable conjectures, the matrix multiplication exponent $\omega=2$ (see \cite[Conjecture~4.7 and Conjecture~3.4]{cohn2005group}), the best provable upper bound we have is $\omega < 2.371552$ by \cite{williams2024new}. 
Nevertheless, given the apparent difficulty of showing $\omega \approx 2$, it is natural to ask whether MM is in fact \emph{hard}.
To that end, we study showing its hardness under the \emph{Strong Exponential Time Hypothesis} (SETH).
However, rather than showing SETH-hardness of MM, we show a \emph{barrier} to proving $n^{\gamma}$-hardness of MM for constant $\gamma > 2$ under SETH. (Because MMV is trivially reducible to MM, our hardness barrier result also applies to MMV.)

We informally define several concepts used in the statement of our result.
SETH says that for large constant $k$, $k$-SAT instances on $n$ variables take nearly $2^n$ time to solve, and the \emph{Nondeterministic Strong Exponential Time Hypothesis} (NSETH) says that certifying that such $k$-SAT formulas are \emph{not} satisfiable takes nearly $2^n$ time even for nondeterministic algorithms.
We call a matrix \emph{rigid} if the Hamming distance between it and all low-rank matrices is high (the Hamming distance and rank are quantified by two parameters).
Rigid matrices have many connections to complexity theory and other areas, and a key goal is to find explicit, deterministic constructions of such matrices.

Intuitively, NSETH rules out showing hardness of problems with non-trivial co-nondeterministic algorithms under SETH. Somewhat more precisely, assuming NSETH, problems contained in $\cc{coTIME}[f(n)]$ (but perhaps only known to be in $\cc{TIME}[g(n)]$ for $g(n) = \omega(f(n))$), cannot be shown to be $\Omega(f(n)^{1+\eps})$-hard under SETH. (See \cref{def:seth-hardness} for a formal definition of SETH-hardness.)
K\"{u}nnemann~\cite{kunnemann2018nondeterministic} noted that, because Freivald's algorithm shows that MMV is in $\cc{coTIME}[n^2 \cdot \poly(\log n)]$, NSETH rules out showing $\Omega(n^{\gamma})$ hardness of MMV under SETH for constant $\gamma > 2$.

In this work, we extend this observation and give a barrier not only to showing SETH-hardness of MMV but to showing hardness of MM.
Our barrier result says that, if there exists a constant $\gamma > 2$ and a reduction from $k$-SAT to MM such that a $O(n^{\gamma - \eps})$-time algorithm for MM for any constant $\eps > 0$ breaks SETH, then either (1) the \emph{Nondeterministic Strong Exponential Time Hypothesis} (NSETH) is false, or (2) a new non-randomized algorithm for computing (arbitrarily large) rigid matrices exists.
We also note that, by known results, falsifying NSETH implies a new circuit lower bound as a consequence.
In short, our barrier result says that showing $n^{\gamma}$-hardness of MM under SETH for $\gamma > 2$ would lead to major progress on important questions in complexity theory.
See \cref{thm:non-seth-hardness} for a formal statement.

A key idea that we use for proving our result is that it is possible to compute the product of two \emph{non-rigid} matrices efficiently using a \emph{nondeterministic} algorithm. This follows from two facts. First, by definition, a non-rigid matrix is the sum of a low-rank matrix $L$ and a sparse matrix~$S$, and using nondeterminism it is possible to guess $L$ and $S$ efficiently. Second, it is possible to compute the product of two sparse matrices or a low-rank matrix and another matrix efficiently.
(In fact, we also use nondeterminism to guess a \emph{rank factorization} of $L$, and this factorization is what allows for fast multiplication by $L$.) 

Very roughly, we prove the barrier result as follows. We first suppose that there is a reduction from $k$-SAT to (potentially multiple instances of) matrix multiplication. In particular, such a reduction outputs several
pairs of matrices to be multiplied.
We then analyze three cases:
\begin{enumerate}
\item If the matrices output by this reduction always have small dimension (as a function of $n$), then we can compute the product of each pair quickly using standard matrix multiplication algorithms (even using na\"{i}ve, cubic-time matrix multiplication). This leads to a fast, deterministic algorithm for $k$-SAT, which refutes SETH (and hence NSETH).

\item \label{item:refuting-NSETH} If the matrices output by this reduction are always \emph{not} rigid, then we can compute the product of each pair quickly using the nondeterministic algorithm sketched above. This leads to a fast, nondeterministic algorithm for showing that $k$-SAT formulas are not satisfiable, which refutes NSETH.

\item \label{item:new-rigid-mat-alg}
Finally, if neither of the above cases holds, then the reduction must sometimes output rigid matrices with large dimension as a function of $n$.
So, we obtain an algorithm for generating arbitrarily large rigid matrices using an $\NP$ oracle: iterate through all $k$-SAT formulas $\varphi$ with at most a given number of variables, apply the reduction from $k$-SAT to MM to each formula, and then use the $\NP$ oracle to check whether each large matrix output by the reduction is rigid.
\end{enumerate}

We remark that although NSETH is a strong and not necessarily widely believed conjecture,~\cite{jahanjou2015local,carmosino2016nondeterministic} showed that refuting it (as in \cref{item:refuting-NSETH} above) would nevertheless imply an interesting new circuit lower bound. Specifically, they showed that if NSETH is false, then the complexity class $\cc{E}^{\NP}$ requires series-parallel circuits of size $\omega(n)$.

Additionally, we remark that despite how slow the ``iterate through all sufficiently large $k$-SAT formulas and apply the $k$-SAT-to-MM reduction to each one'' algorithm described in \cref{item:new-rigid-mat-alg} seems, it would still substantially improve on state-of-the-art non-randomized algorithms for generating rigid matrices. This is also true despite the fact that the algorithm uses an $\NP$ oracle. See \cref{sec:rigidity-comparison} for a more thorough discussion.

\subsubsection{Reductions}
\label{sec:summary-reductions}

Again, motivated by the apparent challenge of fully derandomizing Freivalds's algorithm, we study relationships between variants of MMV with the goal of understanding what makes the problem hard to solve deterministically in $\Ot(n^2)$ time but easy to solve in $\Ot(n^2)$ time using randomness (in contrast to MM).
More specifically, we study which variants are potentially easier than MMV (i.e., reducible to MMV, but not obviously solvable deterministically in $\Ot(n^2)$ time using known techniques), equivalent to MMV, and potentially harder than MMV (i.e., variants to which MMV is reducible, but which are not obviously as hard as MM). 
We study these questions by looking at deterministic $\Ot(n^2)$-time reductions between variants. 
See \cref{fig:web-of-reductions} for a summary of our results.

First, we show that two apparently special cases of MMV are in fact equivalent to MMV.  These special cases are: (1) the \emph{Inverse Verification Problem}, where the goal is to verify that $B = A^{-1}$ for input matrices $A$ and $B$ (equivalently, the special case of MMV where $C = I_n$), and (2) the \emph{Symmetric MMV Problem}, where the input matrices $A$ and $B$ (but not necessarily $C$) are symmetric. See \cref{thm:mmv_to_symMmv}. These reductions are relatively simple, and complement the (also simple) reduction of \cite{kunnemann2018nondeterministic}, who showed that the All Zeroes Problem (i.e., the special case of MMV where $C = 0$) is MMV-complete (see \cref{thm:mmv-all-zeroes-reduction}).

Second, we identify two problems that are $\Ot(n^2)$-time reducible to MMV, but are not clearly solvable in $\Ot(n^2)$ time or equivalent to MMV. These problems are: (1) the \emph{Strong Symmetric MMV Problem}, where all three of the input matrices $A$, $B$, and $C$ are symmetric, and (2) the Monochromatic All Pairs Orthogonal Vectors Problem, where the goal is, given vectors $\vec{a}_1, \ldots, \vec{a}_n$ to decide whether $\iprod{\vec{a}_i, \vec{a}_j} = 0$ for all $i \neq j$. 

Third, we identify two problems for which there are $\Ot(n^2)$-time reductions from MMV and that are $\Ot(n^2)$-time reducible to MM. These ``MMV/MM-intermediate problems'' are: (1) the Matrix Product Sparsity Problem (MPS), in which the goal is, given matrices $A$ and $B$ and $r \geq 0$ as input, to 
decide whether $\norm{AB}_0 \leq r$, and (2) the $k$-MMV problem, in which given matrices $A_1, \ldots, A_k, C$ as input, the goal is to decide whether $\prod_{i=1}^k A_i = C$. 
We note that MPS is equivalent to the counting version of the Orthogonal Vectors Problem ($\#\problem{OV}$).%
\footnote{Indeed, Monochromatic All Pairs Orthogonal Vectors is no harder than MMV (and not obviously equivalent), Bichromatic All Pairs Orthogonal Vectors is equivalent to the All Zeroes Problem and is therefore equivalent to MMV, and MPS/$\#\problem{OV}$ is at least as hard as MMV. In the fine-grained complexity setting OV variants are usually considered with $n$ vectors in dimension $d = \poly(\log n)$; here we are considering the regime where $d = n$.}
We additionally show that $k$-MMV is equivalent to the $k$-All Zeroes problem, i.e., $k$-MMV where $C$ is fixed to be $0$.

\subsection{Related Work}
\label{sec:related-work}

We next discuss other algorithms for MMV and related problems on $n \times n$ integer matrices $A$, $B$, and $C$. We summarize these algorithms, as well as ours, in \cref{tbl:mmv-algs}. 
We start by noting that it suffices to consider the special case of MMV where $C = 0$ (i.e., where the goal is to decide whether $AB = 0$), which is called the All Zeroes Problem. Indeed, a result from~\cite{kunnemann2018nondeterministic} (which we include as \cref{thm:mmv-all-zeroes-reduction}) shows that there is a simple $O(n^2)$-time reduction from MMV on $n \times n$ matrices $A, B, C$ to the All Zeroes problem on $2n \times 2n$ matrices $A', B'$ with the property that $\norm{AB - C}_0 = \norm{A' B'}_0$.
So, for this section we consider the All Zeroes Problem without loss of generality.

Perhaps the most closely related works to ours are \cite{journals/ipl/IwenS09,abboud2024time}, which use fast rectangular matrix multiplication for the \emph{Output-Sensitive Matrix Multiplication Problem} (OSMM). In $t$-OSMM, the goal is, given matrices $A, B$ as input, to compute the product $AB$ when it is promised to be $t$-sparse. There is a trivial reduction from MMV when the output is promised to be $t$-sparse to $t$-OSMM---compute $AB$ and check whether it is equal to $0$. Indeed, OSMM is essentially the search version of sparse MMV. 
However, it is not clear that the measurement matrix $M$ in~\cite{journals/ipl/IwenS09} is deterministically computable in $\Ot(n^2)$ time, and so the algorithm in~\cite{journals/ipl/IwenS09} is a non-uniform algorithm as described. There are other candidate measurement matrices with deterministic constructions that may work for a similar purpose~\cite{Iwen-Personal-23}, but the exact tradeoffs do not seem to have been analyzed and it is not clear that it is possible to get a (uniform) algorithm with the same parameters.
Additionally,~\cite{journals/ipl/IwenS09} only handles the case when all columns or rows of $AB$ are promised to have a given sparsity, rather than the case where there is a ``global bound'' of $t$ on the sparsity of the matrix product itself. 

The main algorithm in~\cite{abboud2024time} for OSMM (summarized in~\cite[Theorem 1.4]{abboud2024time}) runs in randomized time $O(n^{1.3459 \delta})$ when both the input matrices $A, B$ and their product $AB$ are $n^{\delta}$-sparse.%
\footnote{A more general version of this theorem, which gives an algorithm whose running time depends both on the sparsity of the input matrices $A, B$ and of their product $AB$, appears as~\cite[Theorem 1.7]{abboud2024time}.}
For the special case when all entries in $A, B$ are non-negative, ~\cite{abboud2024time} gives a deterministic algorithm with the same running time as their randomized algorithm.
We note that~\cite{abboud2024time} was written independently and concurrently with this work.

Besides simply using matrix multiplication, perhaps the most natural idea for an algorithm for the All Zeroes problem is to compute a random entry $(AB)_{i,j}$ of $AB$ and check whether it is non-zero. If $\norm{AB}_0 \geq n^{\delta}$, then sampling, say, $10 n^{2 - \delta}$ random entries of $AB$ independently will find a non-zero entry with good constant probability. Because computing each such entry amounts to computing an inner product, and sampling indices $i, j \sim \set{1, \ldots, n}$ takes roughly $2 \log_2 n$ random bits, this algorithm overall takes $\Ot(n^{3 - \delta})$ time and $O(n^{2-\delta} \log n)$ random bits.
So, this algorithm is relatively efficient and succeeds with good probability in the case when $AB$ is dense, but even then requires a relatively large number of random bits. We also note the somewhat odd fact that this algorithm is most efficient when $AB$ is dense, whereas our algorithms are most efficient when $AB$ is sparse. 

Freivalds's algorithm~\cite{freivalds1979fast} works by sampling a uniformly random vector $\vec{x} \sim \bit^n$, and outputting ``YES'' if $AB \vec{x} = \vec{0}$ and ``NO'' otherwise. If $AB = 0$, then this algorithm is always correct, and if $AB \neq 0$ then it fails with probability at most $1/2$.%
\footnote{To see this, note that in the latter case some row $\vec{s}^T$ of $AB$ must be non-zero, and let $j^*$ be the index of the last non-zero entry in $\vec{s}$. Then for uniformly random $\vec{x} \sim \bit^n$, $\Pr[AB\vec{x} = \vec{0}] \leq \Pr[\iprod{\vec{s}, \vec{x}} = 0] = \Pr[s_{j^*} x_{j^*} = -\sum_{k = 1}^{j^* - 1} s_k x_k] \leq 1/2$. Moreover, this holds for matrices $A, B$ over any ring $R$, and so Freivalds's algorithm works for MMV over any ring $R$.}
In particular, Freivalds's algorithm has one-sided error with no false negatives (i.e., it is a $\coRP$ algorithm).

A key idea for subsequent algorithms was to reduce MMV to a question about polynomials. The main idea is the following.
Define $\vec{x} := (1, x, x^2, \ldots, x^{n-1})^T$, where $x$ is an indeterminate, and define $p_i(x) := (AB \vec{x})_i = \sum_{j = 1}^n (AB)_{i,j} \cdot x^{j-1}$. Note that $AB = 0$ if and only if the polynomials $p_i(x)$ are identically zero (as formal polynomials) for all $i \in \set{1, \ldots, n}$. Furthermore, if the $i$th row of $AB$ is non-zero then $p_i(x)$ is a non-zero polynomial of degree at most $n - 1$, and therefore has at most $n - 1$ distinct complex (and hence integral) roots.
So, for such $p_i(x)$ and a non-empty set 
$S \subset \Z$, $\Pr_{\alpha \sim S}[p(\alpha) = 0] \leq (n-1)/\card{S}$, which is less than $1/2$ when $\card{S} \geq 2n$.
This observation leads to the following algorithm for MMV, which forms the basis for Kimbrel and Sinha's algorithm~\cite{kimbrel1993probabilistic}. Sample $\alpha \sim \set{1, \ldots, 2n}$, and output ``YES'' if and only if $AB \vec{\alpha} = \vec{0}$ for $\vec{\alpha} := (1, \alpha, \alpha^2, \ldots, \alpha^{n-1})^T$. Using associativity, it is possible to compute this product as $A(B \vec{\alpha})$ using $O(n^2)$ arithmetic operations.

However, there is an issue with this algorithm: it requires computing powers of $\alpha$ up to $\alpha^{n-1}$. These powers require $\Omega(n)$ bits to represent for any integer $\alpha \geq 2$, and so performing arithmetic operations with them takes $\Omega(n)$ time.
To solve this, Kimbrel and Sinha instead consider the ``test vector'' $\vec{\alpha}$ modulo an (arbitrary) prime $2n \leq q \leq 4n$, which they can find deterministically in $O(n^2)$ time.
They show that their algorithm is still correct with good probability (over the choice of $\alpha$) with this modification.

Korec and Wiedermann~\cite{korec2014deterministic} showed how to \emph{deterministically} find a good $\alpha$ for the above test---that is, a value $\alpha$ such that $p_i(\alpha) \neq 0$ if $p_i$ is not identically zero---using \emph{Cauchy's bound}, which gives an upper bound on the magnitude of the largest root of a polynomial as a function of the polynomial's coefficients.
Namely, they just choose $\alpha$ larger than Cauchy's bound. (They note that the maximum magnitude of an entry in $AB$---and hence of a coefficient in any of the polynomials $p_i(x)$---is at most $n \mu^2$, where $\mu$ is the maximum magnitude of an entry in $A$ or $B$.)
Their algorithm uses only $O(n^2)$ arithmetic operations, but again requires computing powers of $\alpha$ up to $\alpha^{n-1}$, and therefore the algorithm has bit complexity $\Omega(n^3)$.

Additionally, we mention the work of K\"{u}nnemann~\cite{kunnemann2018nondeterministic}, which works for MMV over finite fields $\F_q$ with $q > n^2$ (he reduces MMV over the integers to MMV over such fields). His algorithm works by considering the bivariate polynomial $f(x, y) = f_{A, B}(x, y) := \vec{x}^T AB \vec{y}$ for $\vec{x} = (1, x, x^2, \ldots, x^{n-1})$, $\vec{y} = (1, y, y^2, \ldots, y^{n-1})$, where $x$ and $y$ are indeterminates, and the corresponding univariate polynomial $g(x) = g_{A, B}(x) := f(x, x^n)$. The coefficient of $x^{(i-1) + (j-1)n}$ in $g(x)$ (and of $x^{i-1} y^{j-1}$ in $f(x, y)$) is equal to $(AB)_{i,j}$, and so to decide whether $AB = 0$ it suffices to decide whether $g(x)$ (or $f(x, y)$) is identically zero as a formal polynomial.%
\footnote{Indeed,~\cite{kunnemann2018nondeterministic} notes that this mapping from $A, B$ to $g(x)$ is a reduction from the All Zeroes Problem to Univariate Polynomial Identity Testing (UPIT).}
He shows that to do this it in turn suffices to decide whether $g(\alpha^i) = 0$ for all $i \in \set{0, \ldots, t - 1}$, where $\alpha \in \F_q$ is an element of order at least $n^2$ and $t = n^{\delta}$ is an upper bound on the sparsity of $AB$.
Indeed, he notes that the system of equations $g(1) = \cdots = g(\alpha^{t - 1}) = 0$ is a Vandermonde system of homogeneous linear equations in the at most $t$ non-zero entries $(AB)_{i, j}$ in $AB$, and so its only solution is the solution $(AB)_{i,j} = 0$ for all $1 \leq i, j \leq n$ (i.e., it must be the case that $AB = 0$).
To evaluate $g$ on the $t$ values $1, \alpha, \ldots, \alpha^{t - 1}$ quickly, he uses a known result about fast multipoint polynomial evaluation.

We also note that MMV and its variants have also been studied from angles other than derandomization of Freivalds's algorithm. Notably,~\cite{buhrman2004quantum} gave a $O(n^{5/3})$-time \emph{quantum} algorithm for MMV,~\cite{hon2023verifying} studied the \emph{Boolean} Matrix Multiplication Verification problem, and \cite{journals/algorithmica/GasieniecLLPT17,wu2023correcting} study the problem of \emph{correcting} matrix products. I.e., they study the problem of \emph{computing} $AB$ given matrices $A$, $B$, and $C$ where $\norm{AB - C}_0$ is guaranteed to be small, which K\"{u}nnemann showed is equivalent to OSMM.

Finally, we remark that other recent works including~\cite{carmosino2016nondeterministic,belova2023polynomial,aggarwal2023lattice,belova2023computations} have studied ``barriers to SETH hardness'' akin to~\cref{thm:non-seth-hardness}.

\subsection{Open Questions}
\label{sec:open-questions}

Of course, the main question that we leave open is whether Freivalds's algorithm can be fully derandomized, i.e., whether there is a deterministic $\Ot(n^2)$-time algorithm for MMV on $n \times n$ matrices over finite fields $\F_q$ with $q \leq \poly(n)$ and integer matrices with entries $[-n^c, n^c]$ for constant $c > 0$.
Giving such an algorithm for a natural special case of MMV (such as those described in \cref{sec:MMV-easy}) also seems like a natural step in this direction.
Additionally, it would be interesting to extend our results for MMV in the sparse regime to Output Sensitive Matrix Multiplication. The coding-theoretic techniques that we use seem amenable to this.

\full{
\subsection{Acknowledgments}
\label{sec:acknowledgments}

We thank Amir Nayyeri for many helpful discussions in the early stages of work on this paper, and Mark Iwen~\cite{Iwen-Personal-23} for answering questions about~\cite{journals/ipl/IwenS09}.} We also thank the anonymous reviewers for their helpful comments.

\section{Preliminaries}

We define $\norm{\vec{v}}_0$ (respectively, $\norm{M}_0$) to be the number of non-zero entries (i.e., Hamming weight) in a vector $\vec{v}$ (respectively, matrix $M$). We call a vector $\vec{v}$ (respectively, matrix $M$) \emph{$t$-sparse} if $\norm{\vec{v}}_0 \leq t$ (respectively, if $\norm{M}_0 \leq t$).

\subsection{Matrix Multiplication}

We first give definitions related to matrix multiplication.

\begin{definition}[$\textrm{MM}_R$]
The Matrix Multiplication Problem over a ring $R$ ($\textrm{MM}_R$) is defined as follows. Given matrices $A, B \in R^{n \times n}$ for some $n\in\Z^{+}$, compute their product $AB$.
\end{definition}

\begin{definition}[Rectangular and Square Matrix Multiplication Exponents] \label{def:rmm-exponents}
    For $\alpha, \beta, \gamma \in [0,1]$, we define the \emph{rectangular matrix multiplication exponent} $\omega(\alpha, \beta, \gamma)$ to be 
    the infimum over $\omega' > 0$ such that the product of an $n^\alpha \times n^\beta$ matrix $A$ and an $n^\beta \times n^\gamma$ matrix $B$ can be computed using $O(n^{\omega'})$ arithmetic operations.
    We also define the (square) \emph{matrix multiplication exponent} as $\omega := \omega(1,1,1)$.
\end{definition}

Additionally, we define the \emph{dual matrix multiplication exponent} $\dualMMexp$ as
\begin{equation}
\dualMMexp := \sup \set{\omega' > 0 : \omega(1, 1, \omega') = 2} \ \text{.}
\end{equation}

A recent line of work~\cite{gall2018improved,alman2021refined,duan2022faster,gall2023faster,williams2024new} has shown improved bounds on rectangular matrix multiplication exponents $\omega(1, 1, \beta)$ (including $\omega = \omega(1, 1, 1)$) and the dual matrix multiplication exponent $\dualMMexp$. The current best bounds for all of these quantities---including, to the best of our knowledge, $\omega(1, 1, \beta)$ for all $\beta \in (\dualMMexp, 1]$---appear in~\cite{williams2024new}.
In particular,~\cite{williams2024new} proves the following bounds:
\begin{align}
    \omega &\leq 2.371552 \ \text{,} \label{eq:MMexp-UB} \\
    \dualMMexp &\geq 0.321334 \ \text{.} \label{eq:dualMMexp-LB}
\end{align}

We will use the following upper bound on rectangular matrix multiplication exponents $\omega(1, 1, \beta)$ in terms of $\beta$, $\omega$, and $\dualMMexp$.
\begin{theorem}[{\cite{lotti1983asymptotic,le2012faster}}] \label{thm:rectMM-ub}
Let $\beta \in [0, 1]$. Then
    \[
    \omega(1,1,\beta) \leq
\begin{cases}
    2 & \text{if $0 \leq \beta \leq \omega^{\perp}$ ,} \\
    2 + (\omega - 2) \cdot \dfrac{\beta - \omega^{\perp}}{1 - \omega^{\perp}} & \text{if $\omega^{\perp} < \beta \leq 1$ .}
\end{cases}
\]
    
\end{theorem}

We note that the second case of \cref{thm:rectMM-ub} implies that if $\omega > 2$ then $\omega(1, 1, \beta) < \omega$ for $\beta < 1$.
We will also use the following lower bound on rectangular matrix multiplication exponents.

\begin{lemma} \label{lem:omega_a_1_b}
    For any $\alpha, \beta \in [0, 1]$, $\omega(\alpha, 1, \beta) \geq \omega - 2 + \alpha + \beta$.
\end{lemma}

\begin{proof}
    Consider two matrices $X, Y \in R^{n \times n}$ over some ring~$R$.
    Partition $X$ into $\lceil n^{1 - \alpha} \rceil$ submatrices with dimensions $\lceil n^{\alpha} \rceil \times n$, and $Y$ into $\lceil n^{1 - \beta} \rceil$ submatrices with dimensions $n \times \lceil n^{\beta} \rceil$ (padding $X$ and $Y$ with zeroes if $\lceil n^{\alpha} \rceil$ or $\lceil n^{\beta} \rceil$ does not divide $n$).

    Then for any constant $\eps > 0$, the product $XY$ can be computed using $O(\lceil n^{1 - \alpha} \rceil \cdot \lceil n^{1 - \beta} \rceil \cdot n^{\omega(\alpha, 1, \beta) + \eps}) = O(n^{2 + \omega(\alpha, 1, \beta) - \alpha - \beta + \eps})$ algebraic operations by multiplying each of the $O(n^{1-\alpha} \cdot n^{1 -\beta}) = O(n^{2 - \alpha - \beta})$ pairs of submatrices $X'$ of $X$ and $Y'$ of $Y$ in time $O(n^{\omega(\alpha, 1, \beta) + \eps})$.
    It follows that $\omega \leq 2 + \omega(\alpha, 1, \beta) - \alpha - \beta$, and the result follows by rearranging.
\end{proof}

Finally, we will also use the following fact about equivalences among rectangular matrix multiplication exponents.
\begin{theorem}[{\cite{lotti1983asymptotic}}]\label{thm:rectmult}
For any $\beta \in [0, 1]$, $\omega(1,1,\beta) = \omega(1,\beta,1)= \omega(\beta,1,1)$.
\end{theorem}

\subsection{Matrix Multiplication Verification}

We next define Matrix Multiplication Verification and the All Zeroes Problem over arbitrary rings~$R$ and with a sparsity parameter $t$.

\begin{definition}[$\textrm{MMV}_R^t$]
    The Matrix Multiplication Verification Problem over a ring~$R$ with sparsity parameter~$t = t(n) \geq 0$ ($\textrm{MMV}_R^t$) is defined as follows. Given matrices $A, B,C \in R^{n \times n}$ for some $n\in\Z^{+}$ such that $\norm{AB - C}_0 \leq t$ as input, decide whether $AB = C$. 
\end{definition}

\begin{definition}[$\textrm{AllZeroes}_R^t$] \label{def:all-zeroes}
    The All Zeroes Problem over a ring~$R$ with sparsity parameter~$t = t(n) \geq 0$ ($\textrm{AllZeroes}_R^t$) is defined as follows. Given matrices $A, B \in R^{n \times n}$ for some $n\in\Z^{+}$ such that $\norm{AB}_0 \leq t$ as input, decide whether $AB = 0$.
\end{definition}

We note that when $t = n^2$, $t$ does not restrict the input.

We use the following reduction from~\cite{kunnemann2018nondeterministic} with the additional observations that it is sparsity-preserving and works over arbitrary rings. We provide its simple proof for completeness. 

\begin{theorem}[{\cite[Proposition 3.1]{kunnemann2018nondeterministic}}]\label{thm:mmv-all-zeroes-reduction}
For all rings $R$, $n \in \Z^+$, and $0 \leq t \leq n^2$, there is a reduction from $\textrm{MMV}_R^t$ on $n \times n$ matrices to $\textrm{AllZeroes}_R^t$ on $2n \times 2n$ matrices that uses $O(n^2)$ arithmetic operations on $R$. 
\end{theorem}

\begin{proof}
Let the matrices $A,B,C$ be an instance of MMV, and define
\[
A' := \begin{bmatrix}
         A & -I \\
         0 & 0 \\ 
     \end{bmatrix} \ \text{,} \quad
B' := \begin{bmatrix}
         B & 0 \\
         C & 0 \\ 
     \end{bmatrix} \ \text{.}
\]
We then have that
\[
A' B' =
     \begin{bmatrix}
         A & -I \\
         0 & 0 \\ 
     \end{bmatrix}
     \cdot
     \begin{bmatrix}
         B & 0 \\
         C & 0 \\ 
     \end{bmatrix}
     =
     \begin{bmatrix}
         AB-C & 0 \\
         0 & 0 \\ 
     \end{bmatrix} \ \text{,}
\]
and therefore that $\norm{A' B'}_0 = \norm{AB - C}_0$. It follows that $A', B'$ is a YES instance of $\textrm{AllZeroes}_R^t$ if and only if $A, B, C$ is a YES instance of $\textrm{MMV}_R^t$, as needed. \\

\end{proof}

\subsection{Coding Theory}
\label{sec:coding-theory}

We next review useful background information about (linear) error-correcting codes $\C \subseteq \F_q^n$. We refer the reader to \cite{grs-essential-coding-theory} for a comprehensive resource on coding theory.
We start by giving a ``dual'' definition of error-correcting codes in terms of parity-check matrices $H$, which will be useful for our work.

\begin{definition} \label{def:code}
Let $k \geq 0$ and $n \geq 1$ be integers satisfying $k \leq n$, let $q$ be a prime power, and let $H \in \F_q^{(n-k) \times n}$ be a matrix with linearly independent rows. Then
\[
\C = \Cperp(H) := \set{\vec{x} \in \F_q^n : H \vec{x} = \vec{0}}
\]
is the linear error-correcting code with parity check matrix $H$.
\end{definition}

We remark that it is also possible to define a code in the primal way using a \emph{generator matrix} $G \in \F_q^{k \times n}$, setting
\[
\C = \C(G) = \set{G^T \vec{x} : \vec{x} \in \F_q^k} \ \text{.}
\]
That is, $\C(G)$ is the linear subspace of $\F_q^n$ spanned by rows of $G$.
A generator matrix of the form $G = (I_k \stbar G')$ is said to be in \emph{systematic form}.
Although we do not directly need this definition, it is related to a deep connection between coding theory and matrices with the ``regularity'' property we need for \cref{thm:fast-random-intro}; see \cref{rem:mds-vandermonde-cauchy}.

The \emph{minimum distance} of a linear code $\C$ is defined as
\[
d = d(\C) := \min_{\substack{\vec{x}, \vec{y} \in \C, \\ \vec{x} \neq \vec{y}}} \norm{\vec{x} - \vec{y}}_0 = \min_{\vec{x} \in \C \setminus \set{\vec{0}}} \norm{\vec{x}}_0 \ \text{.}
\]
Note that for any $\C \subseteq \F_q^n$, $1 \leq d(\C) \leq n$.

A linear error-correcting code $\C$ with parity-check matrix $H \in \F_q^{(n-k) \times n}$ and minimum distance $d = d(\C)$ is called a $[n, k, d]_q$ code. Here $n$ is the \emph{block length} of the code and $k$ is the \emph{dimension} of the code.

A primary goal of coding theory is to study the largest values of $k = k(n)$ and $d = d(n)$ for which $[n, k, d]_q$ codes exist (either for constant $q$ or a family of values $q = q(n)$).
A fundamental result in coding theory called the \emph{Singleton bound} asserts that for any $[n, k, d]_q$ code,
\begin{equation} \label{eq:singleton}
d \leq n - k + 1 \ \text{.}
\end{equation}
See, e.g.,~\cite[Theorem 4.3.1]{grs-essential-coding-theory}.
A code $\C$ for which the preceding inequality is tight (i.e., for which $d = n - k + 1$) is called a \emph{maximum distance separable (MDS)} code.
We will use the fact that MDS codes exist for a wide range of values of $k$ and $n$, and moreover that these codes have efficiently computable parity check matrices $H = H(k, n)$.
In particular, generalized Reed-Solomon (GRS) codes are a family of such MDS codes (see~\cite[Theorem 5.1.1]{hall-coding} and \cite[Claim 5.2.3]{grs-essential-coding-theory}).

\begin{theorem}[{\cite[Theorem 5.1.1]{hall-coding}}] \label{thm:reed-solomon}
Let $q$ be a prime power, and let $k$ and $n$ be integers satisfying $0 < k \leq n \leq q$. Then there exists a $[k, n, n - k + 1]_q$ code $\C$. Moreover, there exists an algorithm to compute a parity check matrix $H = H(k, n) \in \F_q^{(n-k) \times n}$ of such a code $\C$ in $n (n-k) \cdot \poly(\log q)$ time.
\end{theorem}

We remark that the parity-check matrix $H$ has $(n-k) \cdot n$ entries, and so the preceding theorem says that $H$ is computable in $\poly(\log q)$ time per entry. For GRS codes, $H$ is the transpose of a Vandermonde matrix. Such matrices are defined as follows.

\begin{definition}[Vandermonde Matrix]
    Let $m$ and $n$ be positive integers.
    Let $x_1, \ldots, x_n \in \F$ be distinct elements of a field $\F$.
    A \emph{Vandermonde matrix} $V \in \F^{n \times m}$ is a matrix with $V_{i,j} := x_i^{j-1}$ for distinct elements $x_1, \ldots, x_n \in \F$. I.e., $V$ has the following form: 
    \[
        \begin{bmatrix}
            1 & x_1 & x_1^2 & \cdots & x_1^{m-1} \\
            1 & x_2 & x_2^2 & \cdots & x_2^{m-1} \\
            \vdots & \vdots & \vdots & \ddots & \vdots \\
            1 & x_n & x_n^2 & \cdots & x_n^{m-1}
        \end{bmatrix} \ \text{.}
    \]
\end{definition}

As a consequence of \cref{thm:reed-solomon}, the definition of minimum distance, and the definition of a parity check matrix we get the following corollary. In particular, the corollary uses the fact that if $H$ is a parity check matrix for a code with minimum distance $d$ then no non-zero vector $\vec{x}$ with $\norm{\vec{x}}_0 < d$ can be such that $H\vec{x} = \vec{0}$.
\begin{corollary} \label{cor:mds-parity-check}
Let $q$ be a prime power, and let $k$ and $n$ be integers satisfying $0 < k \leq n \leq q$. Then there exists a deterministic algorithm that runs in $n (n-k) \cdot \poly(\log q)$ time for computing a matrix $H \in \F_q^{(n-k) \times n}$ such that for all non-zero $\vec{x} \in \F_q^n$ with $\norm{\vec{x}}_0 \leq n - k$, $H\vec{x} \neq \vec{0}$. 
\end{corollary}
We note that if $q \leq \poly(n)$ then the algorithm in \cref{cor:mds-parity-check} runs in $\Ot(n^2)$ time.
We also again emphasize that by the Singleton bound (\cref{eq:singleton}) the parameters in \cref{cor:mds-parity-check} are optimal in terms of the trade-off between between $k$ and $d$ for fixed $n$.

\subsection{Regular and Super Regular Matrices}
\label{sec:regular-matrices}

We next define and discuss regular matrices.

\begin{definition}[Regular Matrix]
Let $\F$ be a field and $k\in\Z^+$. A matrix $S\in\F^{n\times m}$ is \emph{$k$-regular} matrix, if all of its $k \times k$ submatrices are non-singular. $A$ is \emph{super regular}, if it is $k$-regular for every $1\leq k\leq \min(n,m)$.
\end{definition}

We will use the following simple construction of super regular matrices.

\begin{definition}[Cauchy Matrix]
Let $q$ be a prime power, $n$ and $k$ be positive integers satisfying $k+n\geq q$, and $x_1,\ldots,x_n,y_1,\ldots,y_k$ be distinct elements of $\F_q$. Then $S\in\F_q^{n\times k}$ defined by $S_{i,j}=1/(x_i-y_j)$ is a \emph{Cauchy matrix}.
\end{definition}

\begin{remark} \label{rem:mds-vandermonde-cauchy}
We again recall the correspondence between MDS codes and super regular matrices: a matrix $G = (I_k \stbar S) \in \F_q^{k \times n}$ is a generator matrix of an MDS code if and only if $G'$ is super regular~\cite{roth1989mds}. In fact, as~\cite{roth1989mds} shows, it turns out that $G = (I_k \stbar S)\in \F_q^{k \times n}$ is the generator matrix of a generalized Reed-Solomon (GRS) code if and only if $S$ is a Cauchy matrix! Indeed, this illustrates a close relationship between GRS codes, Vandermonde matrices (which are generator matrices of GRS codes not in systematic form), and Cauchy matrices (which correspond to the ``non-identity part'' of generator matrices of GRS codes in systematic form). 
\end{remark}

For completeness, we include a proof that Cauchy matrices are super regular.

\begin{proposition}\label{prop:cauchy}
Let $q$ be a prime power, $n$ and $k$ be positive integers satisfying $k+n\geq q$, and $S\in\F_q^{n\times k}$ be a Cauchy matrix. Then $S$ is super regular.
\end{proposition}
\begin{proof}
For $1\leq m\leq \min(n,k)$, consider an arbitrary submatrix $S'\in\F_q^{m\times m}$ of $S$. Note that $S'$ is a Cauchy matrix, too, and $S'=(1/(x'_i-y'_j))_{i,j=1}^{m}$ for distinct $x'_1,\ldots x'_m,y'_1,\ldots,y'_m$. The determinant of $S'$, known as the Cauchy determinant~\cite[Part~VII, Chapter~1, Problem~3]{polya1972problems}, is 
\[
\det(S')=\frac{\prod_{i=2}^m\prod_{j=1}^{i-1}(x'_i-x'_j)(y'_i-y'_j)}{\prod_{i=1}^m\prod_{j=1}^m(x'_i-y'_j)} \;.
\]
Since $x'_1,\ldots,x'_n,y'_1,\ldots,y'_k$ are distinct, $\det(S')\neq 0$, which finishes the proof.
\end{proof}

As a consequence of \cref{prop:cauchy}, we get an efficient construction of (columns of) super regular matrices. 
\begin{corollary}\label{cor:cauchy}
Let $q$ be a prime power, and let $n$, $k$, and $i$ be positive integers satisfying $k+n\geq q$ and $1\leq i\leq k$. Then there exists a deterministic algorithm that runs in time $n \cdot \poly(\log q, \log n)$, and on input $(q,n,k,i)$ outputs the $i$th column $\vec{c}_i$ of a super regular matrix $S = (\vec{c}_1, \ldots, \vec{c}_k) \in \F_q^{n \times k}$.
\end{corollary}

\subsection{Fine-Grained Complexity}
In the $k$-SAT problem, given a $k$-CNF formula $\varphi$, the task is to check if $\varphi$ has a satisfying assignment. The complement of $k$-SAT, the $k$-TAUT problem, is to decide if all assignments to the variables of a given $k$-DNF formula satisfy it. The celebrated Strong Exponential Time Hypothesis (SETH) is now central to the field of fine-grained complexity.

\begin{definition}[Strong Exponential Time Hypothesis (SETH),~\cite{IP99}]
    For every constant $\eps > 0$ there exists $k \in \Z^+$ such that no deterministic algorithm solves $k$-SAT on formulas with $n$ variables in time $2^{(1 - \eps)n}$.
\end{definition}

The field of fine-grained complexity has leveraged SETH to prove tight bounds on the complexity of a large host of problems (see the excellent surveys by Vassilevska Williams~\cite{V15,V18}). We will also use the following stronger hypothesis.

\begin{definition}[Nondeterministic Strong Exponential Time Hypothesis (NSETH),~\cite{carmosino2016nondeterministic}]
    For every constant $\eps > 0$ there exists $k \in \Z^+$ such that no nondeterministic algorithm solves $k$-TAUT on formulas with $n$ variables in time $2^{(1 - \eps)n}$.
\end{definition}

A classical tool used in studies of algorithms for $k$-SAT is the Sparsification Lemma.
\begin{theorem}[Sparsification Lemma~\cite{impagliazzo2001problems}]\label{sparsification_lemma}
    For every $k \geq 3$ and $\lambda > 0$ there exists a constant $c = c(k, \lambda)$ such that every $k$-SAT formula $\varphi$ with $n$ variables can be expressed as $\varphi = \lor_{i=1}^r \psi_{i}$ where $r \leq 2^{\lambda n}$ and each $\psi_i$ is a $k$-SAT formula with at most $cn$ clauses. Moreover, all $\psi_i$ can be computed in $2^{\lambda n}$-time.
\end{theorem}

Now we formally define the notions of fine-grained reductions and SETH-hardness. We start with the definition of fine-grained reductions from~\cite[Definition~6]{V15} with a small modification that specifies the dependence of $\delta$ on $\varepsilon$.
\begin{definition}[Fine-grained reductions]\label{deg:fg}
    Let $P,Q$ be~problems, $p, q \colon \mathbb{Z}_{\geq 0} \to \mathbb{Z}_{\geq 0}$ be non-decreasing functions and ${\delta\colon\R_{>0}\to\R_{>0}}$.
    We say that $(P, p(n))$ \emph{$\delta$-fine-grained reduces} to~$(Q, q(n))$
    and write $(P,p(n)) \le_{\delta} (Q,q(n))$,
    if~for every $\varepsilon>0$ and $\delta=\delta(\eps)>0$, there exists
    an~algorithm~$\mathcal A$ for~$P$ with oracle access to~$Q$,
    a~constant~$d$, a~function $t(n) \colon \mathbb{Z}_{\geq 0} \to \mathbb{Z}_{\geq 0}$, such that on~any instance of~$P$ of size~$n$, the algorithm~$\mathcal A$
    \begin{itemize}
        \item runs in time at~most $d(p(n))^{1-\delta}$;
        \item produces at~most $t(n)$~instances of~$Q$ adaptively: every instance depends on~the previously produced instances
        as~well as~the answers of~the oracle~$Q$;
        \item the sizes $n_i$ of the produced instances satisfy the inequality
        \[\sum_{i=1}^{t(n)}q(n_i)^{1-\varepsilon} \leq d(p(n))^{1-\delta} \, .\]
    \end{itemize}
\end{definition}

It is easy to see that if $(P,p(n)) \leq_\delta (Q,q(n))$, then for every $\eps>0$, an algorithm for $Q$ running in time $O(q(n)^{1-\eps})$ implies an algorithm for $P$ running in time $O(p(n)^{1-\delta})$. Equipped with this definition, we are ready to define SETH-hardness of a problem. SETH-hardness of a problem is a sequence of fine-grained reductions from $k$-SAT for every value of $k$.

\begin{definition}[SETH-hardness]\label{def:seth-hardness}
Let $p \colon \mathbb{Z}_{\geq 0} \to \mathbb{Z}_{\geq 0}$ be a non-decreasing function. We say that a problem $P$ is \emph{$p(n)$-SETH-hard} if there exists a function $\delta\colon\R_{>0}\to\R_{>0}$ and for every $k\in\N$, 
\[
(k\text{-SAT},2^n)\leq_{\delta} (P,p(n)) \,.
\]
\end{definition}
It is now easy to see that if a problem $P$ is $p(n)$-SETH-hard, then any algorithm solving $P$ in time $p(n)^{(1-\varepsilon)}$ implies an algorithm solving $k$-SAT in time $2^{(1-\delta(\eps))n}$ for all~$k$, thus, breaking SETH.

\subsection{Matrix Rigidity and Circuit Lower Bounds}

A celebrated result of Valiant~\cite{V77} from 1977 shows that all linear functions computed by linear circuits of logarithmic depth and linear size can be written as a sum of a sparse and low-rank matrices. Valiant therefore introduced the notion of \emph{rigid} matrices.

\begin{definition}\label{def:rigidity}
Let $\F$ be a field, $A\in\F^{n\times n}$ be a matrix, and $0\leq r\leq n$. The rigidity of $A$ over~$\F$, denoted by $\cR_r^\F(A)$, is the Hamming distance between $A$ and the set of matrices of rank at most~$r$. Formally,
\[
\cR_r^\F(A) := \min_{S\colon \rank(A-S)\leq r}\|S\|_0 \; .
\]
\end{definition}
In other words, a matrix $A$ has rigidity $\cR_r^\F(A)\geq s$ if and only if $A\in\F^{n\times n}$ \emph{cannot} be written as a sum
$$
A = S + L\; ,
$$
where $S\in\F^{n\times n}$ is an $(s-1)$-sparse matrix, and $L\in\F^{n\times n}$ has low rank $\rank(L)\leq r$.

Valiant proved that computing a matrix~$A\in\F^{n\times n}$ of rigidity $\cR_{\eps n}^\F(A)\geq n^{1+\eps}$ for a constant~$\eps$, requires linear circuits of super-linear size and logarithmic depth. The construction of a rigid matrix even in the complexity class $\cc{E}^\NP$ (i.e., an $\cc{E}^\NP$ algorithm that outputs rigid matrices) would give us a long-awaited super-linear circuit lower bound.

\section{Fast Matrix Multiplication Verification for Sparse Matrices}
\label{sec:fast-mmv-algs}

In this section, we prove the main algorithmic results of this work: we present efficient deterministic and randomized algorithms for sparse MMV over the integers and finite fields. Specifically, in \cref{sec:detalg,sec:randalg}, we prove \cref{thm:rmm-deterministic-intro,thm:fast-random-intro}, respectively. In \cref{sec:linealgbar}, we give a simple barrier to extending these techniques to the general case of MMV.

\subsection{Deterministic MMV for Sparse Matrices}\label{sec:detalg}
We present a deterministic algorithm %
that solves MMV over the ring of integers. We do this by first giving an algorithm that solves MMV over finite fields in \cref{thm:rmm-deterministic-formal}, and then extend this result to the ring of integers in \cref{cor:rmm-deterministic-formal-integers}.

We will use the following simple observation relating the sparsity of a matrix with the sparsity of its rows and columns. This observation also appears, e.g., in~\cite{wu2023correcting}.

\begin{lemma} \label{lem:mat-sparse-implies-rowcol-sparse}
Let $M$ be a $t$-sparse matrix. Then $M$ has a non-zero, $\sqrt{t}$-sparse row or column.
\end{lemma}

\begin{proof}
If $M$ has a non-zero, $\sqrt{t}$-sparse row then the claim holds. So, suppose that all non-zero rows of $M$ have more than $\sqrt{t}$ non-zero entries. Then, because $\norm{M}_0 \leq t$, $M$ has fewer than $\sqrt{t}$ non-zero rows. It follows that all non-zero columns of $M$ are $\sqrt{t}$-sparse.
\end{proof}

We now prove our main theorem about deterministic MMV.

\begin{theorem}[Fast deterministic MMV for sparse matrices, formal]
\label{thm:rmm-deterministic-formal}
Let $0 \leq \delta \leq 2$ be a constant, and $\F_q$ be a finite field. 
Then for every constant $\eps>0$, there is a deterministic algorithm for $\textrm{MMV}_{\F_q}^{n^\delta}$  that runs in $n^{\omega(1, 1, \delta/2)+\eps}\cdot\poly(\log{q})$ time.
\end{theorem}
\begin{proof}
From \cref{thm:mmv-all-zeroes-reduction}, it is sufficient to construct a deterministic algorithm solving the $\textrm{AllZeroes}_\F^{n^\delta}$ problem with the same guarantee on the running time. If $q< n$, then we consider $\F=\F_r$, a field extension of $\F_q$, where $n\leq r < n^2$. (Such a field can be constructed in deterministic time $\Ot(n)$~\cite{shoup1990new}). If $q\geq n$, then we simply set $\F=\F_q$. We remark that $AB=0$ over $\F_q$ if and only if $AB=0$ over $\F$, and that all arithmetic operations over $\F$ can be performed in time $\poly(\log{q},\log{n})$ (see, e.g., \cite{shoup2009computational}).
In the following, given two matrices $A,B\in\F^{n\times n}$, with the guarantee that $\norm{AB}_0\leq n^\delta$, we will determine if $AB=0$. 

Let $H\in\F^{n^{\delta/2}\times n}$ be the parity check matrix from \cref{cor:mds-parity-check} with $k = n - n^{\frac{\delta}{2}}$ computed in time $n^2 \cdot \poly(\log{q})$. Our algorithm reports that $AB=0$ if and only if both $HAB = 0$ and $H(AB)^T = 0$. 

Since the matrix $H$ is computed in time $n^2 \cdot \poly(\log{q},\log{n})$, and computing the matrix products $HAB$ and $H(AB)^T$ takes time $n^{\omega(1, 1, \delta/2)+\eps}\cdot\poly(\log{q})$ for any constant $\eps>0$, we conclude the desired bound on the running time of the algorithm.

It remains to prove the correctness of the algorithm.
First, if $AB=0$, then $HAB=H(AB)^T=0$ for all matrices $H$, and, therefore,  our algorithm correctly concludes that $AB=0$. 
Now, suppose that $AB \neq 0$.
If a non-zero column $\vec{c}$ of $AB$ has $\norm{\vec{c}}_0\leq n^{\delta/2}$ then $H\vec{c}\neq\vec{0}$ by \cref{cor:mds-parity-check} and so $HAB\neq0$.
Similarly, if a non-zero row $\vec{r}$ of $AB$ has $\norm{\vec{r}}_0\leq n^{\delta/2}$ then $H\vec{r}^T\neq\vec{0}$ and so $H(AB)^T \neq 0$.
Moreover, \cref{lem:mat-sparse-implies-rowcol-sparse} implies that $AB$ has a non-zero, $n^{\delta/2}$-sparse row or column, and so the theorem follows.

\end{proof}

We now show how to extend \cref{thm:rmm-deterministic-formal} to the case of MMV over the integers.
\begin{figure}[t]
    \begin{wbox}
    \begin{enumerate}[(1)]
    \item Let $k := \ceil{\sqrt{t}}$.

    \item Compute an arbitrary prime $p$ that satisfies $n \leq p \leq 2n$.
    
    \item Compute a (parity check) matrix $H \in \F_p^{k \times n}$
    as in \cref{cor:mds-parity-check}.

    \item Output YES if $H(AB -C) = 0$ and $H(AB - C)^T = 0$ (where arithmetic is performed over $\Z$), and output NO otherwise.
\end{enumerate}
    \end{wbox}
\caption{\textsc{Deterministic Algorithm for $\textrm{MMV}_{\Z}^t$.}}
\label{fig:det-mmv-alg}
\end{figure}

\begin{corollary}\label{cor:rmm-deterministic-formal-integers}
Let $0 \leq \delta \leq 2$ be a constant. Then for every $\eps>0$, there is a deterministic algorithm (given in \cref{fig:det-mmv-alg}) for $\textrm{MMV}_\Z^{n^\delta}$ for matrices with entries from $[-M,M]$ that runs in $n^{\omega(1, 1, \delta/2)+\eps}\cdot\poly(\log{M})$ time.
\end{corollary}
\begin{proof}
The algorithm follows the high level idea of the algorithm from \cref{thm:rmm-deterministic-formal} with a few modifications that we describe below. We present an algorithm for the $\textrm{AllZeroes}_\Z^{n^\delta}$ problem, and use \cref{thm:mmv-all-zeroes-reduction} to conclude an algorithm for the $\textrm{MMV}_\Z^{n^\delta}$ problem.
Let $p$ be an arbitrary prime satisfying $n\leq p<2n$. Such prime can be found in deterministic time $\Ot(n)$ using the Sieve of Eratosthenes (see, e.g., \cite[Theorem~18.10~(ii)]{modernalgebra}). Let $H\in\F_p^{n^{\delta/2}\times n}$ be the parity check matrix from \cref{cor:mds-parity-check} with $k=n-n^{\delta/2}$ computed in time $n^2\cdot\poly(\log{n})$. Our algorithm reports that $AB=0$ if and only if $HAB=H(AB)^T=0$, where the multiplication is over the integers. 

The running time of the algorithm is dominated by the time required to compute rectangular matrix multiplication, $n^{\omega(1, 1, \delta/2)+\eps}\cdot\poly(\log{M})$ for any $\eps>0$. If $AB=0$, then $HAB=H(AB)^T=0$ for every $H$. Now, if $AB\neq0$, we can conclude that $AB$ contains a non-zero column or row vector of sparsity at most $n^{\delta/2}$. Let $\vec{c}\in\Z^n$ be a nonzero column of $AB$ of sparsity $\norm{\vec{c}}_0\leq n^{\delta/2}$. Assume that $H\vec{c}=\vec{0}$ where the multiplication is over the integers. We factor out the largest power of $p$ dividing all coefficients of $\vec{c}$, and obtain a non-zero vector $\vec{c}'\in\F_p^n$, $\norm{\vec{c}'}_0\leq n^{\delta/2}$ such that $H\vec{c}'=\vec{0}$ where the multiplication is over $\F_p$. This contradicts the guarantee of \cref{cor:mds-parity-check}. Therefore, $H\vec{c}\neq\vec{0}$, and our algorithm detects that $AB\neq C$ in this case. Similarly, in the case when $AB$ contains a non-zero row of sparsity at most $n^{\delta/2}$, the algorithm detects that $AB\neq0$ from $H(AB)^T\neq 0$.
\end{proof}

We again remark that the algorithm from \cref{cor:rmm-deterministic-formal-integers} outperforms the currently best known algorithm from~\cite{kunnemann2018nondeterministic} for $\text{MMV}_\Z^{n^\delta}$ for every $\delta \geq 1.055322$. Indeed, the algorithm from~\cite{kunnemann2018nondeterministic} runs in time $\Ot({n^{1+\delta}})$, while the running time of the algorithm from \cref{cor:rmm-deterministic-formal-integers} is $\Ot\left(n^{\omega(1, 1, \delta/2)+\eps}\right)$ for any $\eps>0$. From Table~1 in~\cite{williams2024new}, the latter running time is asymptotically better than the former one for every $\delta/2>0.527661$.

If the matrix multiplication exponent $\omega=2$, then the MMV problem can be solved in time $n^{2+\eps}$ for every $\eps>0$. Below we show that even if $\omega>2$, then our algorithm solves $\textrm{MMV}_\Z^{n^\delta}$ in better-than-$n^\omega$ time for every constant $\delta<2$.

\begin{corollary}\label{thm:fast_ver}
If $\omega > 2$, then for every $0 \leq \delta < 2$ there exists an $\alpha>0$, and a deterministic algorithm for $\textrm{MMV}_\Z^{n^\delta}$ for matrices with entries from $[-M,M]$ that runs in $n^{\omega-\alpha}\cdot\poly(\log{M})$ time.
\end{corollary}

\begin{proof}
The running time of the algorithm from \cref{cor:rmm-deterministic-formal-integers} is $n^{\omega(1, 1, \delta/2)+\eps}$ for any $\eps>0$. 

If $\delta\leq2\omega^{\perp}$, then by \cref{thm:rectMM-ub}, $\omega(1,1,\delta/2)=2<\omega$. Setting $\eps=\alpha=(\omega-2)/2>0$, we have that the running time of the algorithm is bounded from above by
\[
n^{2+\eps}\cdot\poly(\log{M})=n^{\omega-\alpha}\cdot\poly(\log{M})\,.
\]

If $\delta>2\omega^{\perp}$, then by \cref{thm:rectMM-ub}, 
\[
\omega(1,1,\delta/2)\leq 2 + (\omega - 2)\dfrac{\delta/2 - \omega^{\perp}}{1 - \omega^{\perp}}
\,.\]
In this case, setting 
\[
\alpha=\eps=\dfrac{(\omega - 2)(1 - \delta/2)}{2(1 - \omega^{\perp})}>0
\,,\]
we have that the algorithm runs in time 
\[
n^{\omega(1, 1, \delta/2)+\eps}\cdot\poly(\log{M})=n^{\omega-\alpha}\cdot\poly(\log{M})\,. \qedhere
\]
\end{proof}

Interestingly, \cref{thm:fast_ver} reduces the problem of constructing faster-than-$n^\omega$ algorithms for AllZeroes to the following gap problem: Distinguish the case where the product of two matrices has $\Omega(n^{2-\eps})$ non-zero entries and the case where the product is identically zero. As noted in \cref{sec:related-work}, this leads to a surprising situation: on the one hand, for every $\eps > 0$, if the product of the two matrices is $n^{2-\eps}$-sparse, then there is an $n^{\omega-\alpha}$-time deterministic algorithm for $\alpha>0$; on the other hand, if the product is $\Omega(n^{2-\eps})$-dense, then there is randomized algorithm that runs in time $\Ot(n^{1+\eps})$ and reads only $\Ot(n^{1+\eps})$ entries of the input matrices.

\subsection{Randomized MMV for Sparse Matrices}\label{sec:randalg}
In this section, we extend our framework to the setting of randomness-efficient probabilistic algorithms for MMV. 
We combine our results conditioned on sparsity along with \cite{kimbrel1993probabilistic} to get the best known randomized algorithm for MMV in both running time, and random bits used. In \cref{thm:randomized}, we present an algorithm for MMV over finite fields, and in \cref{cor:randomized} we extend it to the case of matrices over the integers.%

\begin{theorem}[Fast randomized MMV for sparse matrices, formal]\label{thm:randomized}
Let $0 \leq \delta \leq 2$ be a constant, and $\F_q$ be a finite field. Then for every $\eps>0$, there is a randomized algorithm for $\textrm{MMV}_{\F_q}^{n^\delta}$ that runs in $n^2\cdot\poly(\log(nq/\eps))$ time, succeeds with probability $1 - \eps$, and uses at most $\ceil{\delta/2 \cdot \log_2(n) + \log_2(1/\eps)}$ bits of randomness.
\end{theorem}
\begin{proof}
Let $k=\ceil{n^{\delta/2}/\eps}$. If $q< k+n$, then we consider $\F=\F_r$, a field extension of $\F_q$, where $k+n\leq r < (k+n)^2$, otherwise we set $\F=\F_q$. We note that $\F$ can be constructed in time $q\cdot\poly(\log(kq))=n^2\cdot\poly(\log(nq/\eps))$~\cite{shoup1990new}, all arithmetic operations over $\F$ can be performed in time $\poly(\log|\F|)=\poly(\log(nq/\eps))$~\cite{shoup2009computational}, and that $AB=C$ over $\F$ if and only if $AB=C$ over $\F_q$.

Let $S\in\F^{n\times k}$ be an $n^{\delta/2}$-regular matrix, and $\vec{s}\in\F^n$ be a uniformly random column of $S$. Our algorithm reports that $AB=C$ if and only if $AB\vec{s}=C\vec{s}$ and $(AB)^T\vec{s}=C^T\vec{s}$. In the following, we will show that the running time of the algorithm is $n^2\cdot\poly(\log(nq/\eps))$, the number of random bits used is $\ceil{\delta/2 \cdot \log_2(n) + \log_2(1/\eps)}$, and finally we will prove the correctness of the algorithm.

From \cref{cor:cauchy}, one can compute $\vec{s}$ in deterministic time $n\cdot\poly(\log(nq/\eps))$. Each of $AB\vec{s}$ and $(AB)^T\vec{s}$ can be computed by two vector-matrix multiplications in deterministic ${n^2\cdot\poly(\log(nq/\eps))}$ time.

The described algorithm only uses randomness to generate a uniformly random column index $1\leq i\leq k$, which requires $\ceil{\log_2(k)}=\ceil{\delta/2 \cdot \log_2(n) + \log_2(1/\eps)}$ random bits.

If $AB=C$, then for every vector $\vec{s}$, we have $AB\vec{s}=C\vec{s}$ and $(AB)^T\vec{s}=C^T\vec{s}$, and the algorithm correctly concludes that $AB=C$. For the case when $AB\neq C$, we want the algorithm to detect this with probability at least $1-\eps$. To this end, we will show that for every $n^{\delta/2}$-regular matrix $S\in\F^{n\times k}$, either the matrix $ABS$ or the matrix $(AB)^TS$ has at least $1-\eps$ fraction of non-zero columns. 

By \cref{lem:mat-sparse-implies-rowcol-sparse}, if $\norm{AB-C}_0\leq n^\delta$, then $AB-C$ contains a non-zero row or column of sparsity at most $n^{\delta/2}$.
Let us first consider the case where $AB-C$ has a non-zero row $\vec{r}\in\F^n$ of sparsity at most $n^{\delta/2}$. 
We will show that $\vec{r}S$ has fewer than $n^{\delta/2}$ zero coordinates. This will finish the proof, as the success probability of our algorithm is at least the probability of sampling a column where $\vec{r}S$ is non-zero, which is at least $1-n^{\delta/2}/k\geq 1-\eps$.

Assume, towards a contradiction, that $\vec{r}S$ has at least $n^{\delta/2}$ zero coordinates. Let $I\subseteq[n], |I|= n^{\delta/2}$ be an arbitrary set containing (the indices of) all non-zero coordinates of $\vec{r}$, and let $J\subseteq[n], |J|=n^{\delta/2}$ be a set of indices of some $n^{\delta/2}$ zero coordinates of $\vec{r}S$. Then the matrix $S$ restricted to the rows $I$ and columns $J$ is singular (the linear combination of these rows, given by the non-zero vector~$\vec{r}$, is $\vec{0}$). On the other hand, since $S$ is $n^{\delta/2}$-regular, $S$ restricted to $I$ and $J$ is non-singular, which leads to a contradiction. Similarly, in the case when $AB-C$ has a non-zero column $\vec{c}\in\F^n$ of sparsity at most $n^{\delta/2}$, our algorithm will detect that $AB\neq C$ with probability at least $1-\eps$ by checking $(AB)^T\vec{s}=C^T\vec{s}$.
\end{proof}

Below we show how to extend the algorithm of \cref{thm:randomized} to the case of integer matrices.
\begin{figure}[t]
      \centering
     \begin{wbox}
    \begin{enumerate}[(1)]
    \item Let $k := \ceil{\sqrt{t}/\eps}$.
    \item Compute an arbitrary prime $p$ satisfying $k+n \leq p <2(k+n)$.
    \item Compute a $\ceil{\sqrt{t}}$-regular matrix $S = (\vec{s}_1, \ldots, \vec{s}_k) \in \F^{n \times k}$ as in \cref{cor:cauchy}.
    \item Sample a uniformly random column index $i \sim \set{1, \ldots, k}$.
    \item Output YES if $AB\vec{s}_i = C\vec{s}_i$ and $(AB)^T\vec{s}_i = C^T\vec{s}_i$ (where arithmetic is performed over~$\Z$), and output NO otherwise.
\end{enumerate}
    \end{wbox}
    \caption{\textsc{Randomized Algorithm for $\textrm{MMV}_\Z^{t}$.}}
    \label{fig:rand-mmv-alg}
\end{figure}

\begin{corollary}\label{cor:randomized}
Let $0 \leq \delta \leq 2$ be a constant. Then for every 
$\eps>0$, there is a randomized algorithm (given in \cref{fig:rand-mmv-alg}) for $\textrm{MMV}_\Z^{n^\delta}$ for matrices with entries from $[-M,M]$ that runs in $n (n + 1/\eps) \cdot \poly(\log(nM/\eps))$ time, succeeds with probability $1 - \eps$, and uses at most $\ceil{\delta/2 \cdot \log_2(n) + \log_2(1/\eps)}$ bits of randomness. In particular, when $\eps \geq 1/n$, the algorithm runs in $n^2 \cdot \poly(\log n, \log M)$ time.
\end{corollary}

\begin{proof}
Let $k=\ceil{n^{\delta/2}/\eps}$, and $p$ be an arbitrary prime~$p$ satisfying $k+n\leq p< 2(k+n)$. Let $S\in\F_p^{n\times k}$ be an $n^{\delta/2}$-regular matrix, and $\vec{s}\in\F_p^n$ be a uniformly random column of $S$. Our algorithm reports that $AB=C$ if and only if $AB\vec{s}=C\vec{s}$ and $(AB)^T\vec{s}=C^T\vec{s}$, where the multiplication is over $\Z$. 

Using the Sieve of Eratosthenes (see, e.g., \cite[Theorem~18.10~(ii)]{modernalgebra}), one can find a prime~$p$ satisfying $k+n\leq p<2(k+n)$ in deterministic time $\Ot(k+n)\leq (n/\eps) \cdot \poly(\log(n/\eps))$. From \cref{cor:cauchy}, $s$ can be computed in deterministic time $n\cdot \poly(\log(n/\eps))$. Finally $AB\vec{s}$ and $(AB)^T\vec{s}$ can be computed in time $n^2\cdot\poly(\log(nM/\eps))$. The algorithm uses $\ceil{\log_2(k)}=\ceil{\delta/2 \cdot \log_2(n) + \log_2(1/\eps)}$ random bits to generate a uniformly random column index $1\leq i\leq k$.

If $AB=C$, then for every vector $\vec{s}$, we have $AB\vec{s}=C\vec{s}$ and $(AB)^T\vec{s}=C^T\vec{s}$, and the algorithm correctly concludes that $AB=C$. Assume that $AB\neq C$. Since $\norm{AB-C}_0\leq n^\delta$, $AB-C$ contains a non-zero row or column of sparsity at most $n^{\delta/2}$. We will assume that $AB-C$ has a non-zero row $\vec{r}\in\Z^n$ of sparsity at most $n^{\delta/2}$ (the other case is analogous to this one). 
Assume, towards a contradiction, that $\vec{r}S$ has at least $n^{\delta/2}$ zero coordinates. Let $I\subseteq[n], |I|= n^{\delta/2}$ be an arbitrary set containing (the indices of) all non-zero coordinates of $\vec{r}$, and let $J\subseteq[n], |J|=n^{\delta/2}$ be a set of indices of some $n^{\delta/2}$ zero coordinates of $\vec{r}S$. Then the matrix $S$ restricted to the rows $I$ and columns $J$ is singular over the integers. On the other hand, since $S$ is $n^{\delta/2}$-regular, $S$ restricted to $I$ and $J$ is non-singular over $\F_p$, and, thus, over $\Z$. This leads to a contradiction. Therefore, $\vec{r}S$ has at most $n^{\delta/2}$ zero coordinates, and $ABS$ contains at least $k-n^{\delta/2}\geq (1-\eps)k$ non-zero columns. 
\end{proof}

\subsection{A Barrier for Linear Algebraic Approaches to MMV} \label{sec:linealgbar}

In this section, we give a barrier for algorithms for the AllZeroes problem that are based on ``zero tests.'' On an input instance $A, B$, such zero tests compute the matrix product $L (A B) R$ for fixed $L \in \F^{n^{\alpha} \times n}$ and $R \in \F^{n \times n^{\beta}}$ with $\alpha, \beta \in [0, 1]$ and report ``$AB = 0$'' if and only if $L (A B) R = 0$.
In fact, we consider algorithms based on $k$ such tests, i.e., algorithms that report ``$AB = 0$'' if and only if $L_i (A B) R_i = 0$ for fixed $L_i \in \F^{n^{\alpha_i} \times n}$ and $R_i \in \F^{n \times n^{\beta_i}}$ with $\alpha_i, \beta_i \in [0, 1]$ for $i = 1, \ldots, k$. This approach captures a large class of algorithms for MMV, including our deterministic algorithm presented in \cref{fig:det-mmv-alg}.

In \cref{thm:lin-alg-barrier}, we show a barrier to getting a fast algorithm using this approach to MMV.
The idea behind the barrier is as follows. 
Treating the $n^2$ entries $C_{i, j}$ of $C = AB$ as variables, the $i$th zero test $L_i AB R_i$ for fixed $L_i \in \F^{n^{\alpha_i} \times n}, R_i \in F^{n \times n^{\beta_i}}$ corresponds to a system of $n^{\alpha_i + \beta_i}$ homogeneous linear equations. So, for $AB = C = 0$ to be the unique solution to this system of equations, we must have $\sum_{i=1}^k n^{\alpha_i + \beta_i} \geq n^2$.

\begin{theorem}\label{thm:lin-alg-barrier}
    Let $\F$ be a field, $k$ be an integer, and for every $1\leq i\leq k$, $L_i\in\F^{n^{\alpha_i} \times n}, R_i\in\F^{n \times n^{\beta_i}}$ be matrices such that $\rank(L_i)=n^{\alpha_i}$ and $\rank(R_i)=n^{\beta_i}$ for $0\leq\alpha_i,\beta_i\leq1$. Assume that for all $A,B\in\F^{n\times n}$, $AB=0$ if and only if $L_i AB R_i = 0$ for all $1\leq i\leq k$. Then 
    \begin{align*}
    &\sum_{i=1}^k n^{\alpha_i + \beta_i}\geq n^2 \,,\text{ and}\\
    &\sum_{i=1}^k n^{\omega(1, 1, \min(\alpha_i, \beta_i))} \geq \sum_{i=1}^k n^{\omega(\alpha_i,1,\beta_i)}\geq n^{\omega}\,.
    \end{align*}
\end{theorem}
\begin{proof}
Let $C = AB$, and let us view each of the $n^2$ entries $C_{i,j}$ of~$C$ as a variable. Then each constraint $L_i C R_i=0$ induces a homogeneous system of $n^{\alpha_i + \beta_i}$ linear equations over these variables, and the constraints together induce a system of $\sum_{i=1}^k n^{\alpha_i + \beta_i}$ such equations.
In order for a homogeneous system of linear equations to have $C = 0$ (the all-zeroes solution) as its only solution, the number of equations needs to be at least $n^2$. Hence, we have that $\sum_{i=1}^k n^{\alpha_i + \beta_i}\geq n^2$.
Furthermore, by \cref{lem:omega_a_1_b}, $\omega(\alpha_i,1,\beta_i)\geq\omega-2+\alpha_i+\beta_i$. This implies that
\[
\sum_{i=1}^k n^{\omega(1, 1, \min(\alpha_i, \beta_i))}
\geq \sum_{i=1}^k n^{\omega(\alpha_i,1,\beta_i)}
\geq  \sum_{i=1}^k n^{\omega-2+\alpha_i+\beta_i} 
\geq n^{\omega-2} \cdot \sum_{i=1}^k n^{\alpha_i+\beta_i}
\geq n^\omega\,,
\]
which finishes the proof.
\end{proof}

From \cref{thm:lin-alg-barrier}, we get the following corollary, which says that any algorithm for MMV using ``zero tests'' and computing the matrix each of the matrix products $L_i AB R_i$ independently in $\Omega(n^{\omega(1, 1, \min(\alpha_i, \beta_i))})$ time runs in $\Omega(n^{\omega})$ time. We note that $\Omega(n^{\omega(1, 1, \min(\alpha_i, \beta_i))})$ lower bounds the minimum worst-case cost of multiplying (1) a $n^{\alpha_i} \times n$ and an $n \times n$ matrix and (2) an $n \times n$ matrix and an $n \times n^{\beta_i}$ matrix.
Any parenthesization of $L_i AB R_i$ (e.g., $L_i (A (B R_i))$) requires computing a matrix product either of form (1) or (2).

\begin{corollary} \label{cor:lin-alg-algo-barrier}
Let $\alg$ be an algorithm for $\text{MMV}_{\F}$ that behaves in the following way.
On input $A, B \in \F^{n \times n}$, $\alg$ computes matrices $L_1, R_1, \ldots, L_k, R_k$ with $L_i \in \F^{n^{\alpha_i} \times n}, R_i \in \F^{n \times n^{\beta_i}}$ for $0 \leq \alpha_i, \beta_i \leq 1$ such that $AB=0$ if and only if $L_i AB R_i = 0$ for all $1\leq i\leq k$. It then computes each product $L_i AB R_i$ independently using $\Omega(n^{\omega(1, 1, \min(\alpha_i, \beta_i))}) \geq \Omega(n^{\omega(\alpha_i, 1, \beta_i)})$ time.
Then $\alg$ runs in $\Omega(n^{\omega})$ time.
\end{corollary}

We note that \cref{cor:lin-alg-algo-barrier} does \emph{not} rule out more clever ways of computing the matrix products $L_i AB R_i$, such as exploiting the structure of $L_i, R_i$ or computing multiple products simultaneously. However, a stronger version of \cref{cor:lin-alg-algo-barrier} with a more subtle proof does seem to rule out some of these algorithm variants too.

\section{On the non-SETH-hardness of Matrix Multiplication}

In this section, we show a barrier for proving a lower bound for Matrix Multiplication (and, thus, for MMV) under SETH. More specifically, we show that such a result would imply either super-linear circuit lower bounds, or an explicit construction of rigid matrices.

\subsection{A Nondeterministic Algorithm for Multiplying Non-Rigid Matrices}

We will use the following algorithm for multiplying sparse matrices.

\begin{theorem}[{\cite[Theorem~3.1]{yuster2005fast}}]\label{thm:fast_sparse}
    Let $\F$ be a field, $0\leq s\leq 2$, and $A, B \in \F^{n \times n}$ be $n^s$-sparse matrices. If $\omega>2$, then there is a deterministic algorithm that computes $AB$ using 
    \[
    n^{(2s(\omega-2)+2-\omega\dualMMexp)/(\omega-1-\dualMMexp)+o(1)}
    \]
    arithmetic operations over $\F$.
\end{theorem}

For every $\gamma\geq2$, we define the largest matrix sparsity and the largest dimension of rectangular matrices, such that matrix product can be computed in time $n^{\gamma}$.
\begin{definition}
   For every $\gamma \geq 2$, let 
   \begin{align*}
   r(\gamma)&:=\sup\set{r > 0 : \omega(1, 1, r) \leq \gamma} \,,\\
   s(\gamma)&:=\sup\set{s>0 : \text{product of $n^s$-sparse matrices can be computed in $n^\gamma$ field operations}}\,.
   \end{align*}
\end{definition}

Assuming $\omega>2$, for all $\gamma>2$, from \cref{thm:rectMM-ub} and \cref{thm:fast_sparse} we have that
\begin{align}
r(\gamma) &\geq \frac{(\gamma - 2)(1 - \dualMMexp)}{\omega - 2}+ \dualMMexp
\,,\label{eq:rbound}\\
s(\gamma) &\geq \frac{\gamma(\omega-1-\dualMMexp)+\omega\dualMMexp-2}{2(\omega-2)}
\,.\label{eq:sbound}
\end{align}

We now present a simple \emph{non-deterministic} algorithm that efficiently multiplies non-rigid matrices.

\begin{lemma}\label{lem:nondet}
Let $\F_q$ be a finite field, $\gamma\geq 2$ be a constant, and let $A,B\in\F_q^{n\times n}$ be two matrices satisfying 
\begin{align*}
\cR_{n^{r(\gamma)}}^{\F_q}(A) \leq n^{s(\gamma)} \,, \quad
\cR_{n^{r(\gamma)}}^{\F_q}(B) \leq n^{s(\gamma)}
\,.
\end{align*}
Then for every $\alpha>0$, there is a non-deterministic algorithm that computes $AB$ in time $n^{\gamma+\alpha}\cdot\poly(\log{q})$.     
\end{lemma}
\begin{proof}
    First, we non-deterministically guess decompositions of the non-rigid matrices $A$ and $B$ into the sum of a low-rank matrix and a sparse matrix, and then guess a rank factorization of the two low-rank matrices. More specifically, we non-deterministically guess a tuple $(L_A,R_A,S_A,L_B,R_B,S_B)$ such that 
    \begin{itemize}
    \item $A=L_A\cdot R_A+S_A$ and $B=L_B\cdot R_B+S_B$;
    \item $L_A,L_B\in\F_q^{n\times n^{r(\gamma)}}$, and $R_A,R_B\in\F_q^{n^{r(\gamma)}\times n}$;
    \item $S_A,S_B\in\F_q^{n\times n}$, $S_A$ and $S_B$ are $n^{s(\gamma)}$-sparse.
    \end{itemize}
    By the definition of rigidity (\cref{def:rigidity}), such decompositions exist for all matrices $A,B$ satisfying the premise of the lemma. To verify the decomposition, we need to
    \begin{itemize}
    \item compute $L_A\cdot R_A$ and $L_B \cdot R_B$ in deterministic $n^{\gamma+\alpha}\cdot\poly(\log{q})$ time;
    \item verify that $S_A$ and $S_B$ are $n^{s(\gamma)}$-sparse in deterministic $n^2\cdot\poly(\log{q})$ time;
    \item check that $A=L_A\cdot R_A+S_A$ and $B=L_B\cdot R_B+S_B$ in deterministic $n^2\cdot\poly(\log{q})$ time.
    \end{itemize}

    Now we will present a deterministic algorithm that runs in time $n^{\gamma+\alpha}\cdot\poly(\log{q})$ and computes the product 
    \[
    AB=(L_A\cdot R_A+S_A)(L_B\cdot R_B+S_B)=L_A\cdot R_A \cdot L_B\cdot R_B + L_A \cdot R_A \cdot S_B + S_A \cdot L_B \cdot R_B + S_A \cdot S_B
    \,.
    \]
    Recall from \cref{thm:rectmult} that $\omega(r(\gamma),1,1) = \omega(1,r(\gamma), 1) = \omega(1,1,r(\gamma))$. Then the first three terms in the sum above can be computed using rectangular matrix multiplication in time 
    \[
    n^{\omega(1,1,r(\gamma))+\alpha}\cdot\poly(\log{q}) \leq n^{\gamma+\alpha} \cdot\poly(\log{q})\,.
    \]
    Finally, since $S_A$ and $S_B$ are $n^{s(\gamma)}$-sparse, the term $S_A \cdot S_B$ can be computed in time $n^{\gamma+\alpha}\cdot\poly(\log{q})$.
    This finishes the proof of the lemma.
\end{proof}

\subsection{The Main SETH-Hardness Barrier Result}

We are now ready to present the main result of this section: a barrier to proving SETH-hardness of Matrix Multiplication (and hence also of Matrix Multiplication Verification). Below, by $\QP=\DTIME[n^{\poly(\log n)}]$ we denote the class of deterministic algorithms running in quasi-polynomial time.

\begin{theorem} \label{thm:non-seth-hardness}
Let $q$ be a prime power. If Matrix Multiplication over $\F_q$ is $n^\gamma\cdot\poly(\log{q})$-SETH-hard for a constant $\gamma>2$, then one of the following holds:
        \begin{itemize}
            \item
            NSETH is false. %
            \item There is a $\QP^{\NP}$ algorithm ${\cal A}$ such that for infinitely many values of $N$, ${ \cal A}(1^N)$ outputs a matrix $A\in\F_q^{M\times M}$ satisfying 
            \[
            \cR_{M^{r(\gamma')}}^{\F_q}(A)\geq M^{s(\gamma')}
            \] 
            for any $\gamma'<\gamma$, where $M=N^{\Theta(1)}$.
            Specifically, the running time of the algorithm $\cal{A}$ on input $1^N$ is $N^{O(\log\log{N})}$.
        \end{itemize}
\end{theorem}

\begin{proof}
    The high-level idea of the proof is the following. Assume there is a fine-grained reduction~$R_k$ from $k$-SAT to $\textrm{MM}_{\F_q}$ for every $k$. For every $k$-SAT formula on $n$ variables, $R_k$ adaptively produces a (possibly exponentially large) sequence of MM instances. If for all large enough~$n$, all the produced matrices are non-rigid, then from \cref{lem:nondet} we have an efficient non-deterministic algorithm for all encountered instances of MM, and, thus, an efficient non-deterministic algorithm for $k$-TAUT. This refutes NSETH.
    On the other hand, if for infinitely many values of~$n$, at least some of the produced instances of MM are rigid, then we have an algorithm that brute forces all $k$-SAT instances, applies $R$ to them, and then uses the $\NP$ oracle to verify if they are rigid. We will show that this algorithm finds rigid matrices $M\in\F_q^{N\times N}$ in time $N^{O(\log\log N)}$ with an $\NP$ oracle. Below we formalize this argument.

    Assume that $\textrm{MM}_{\F_q}$ is $n^\gamma$-SETH-hard (see \cref{def:seth-hardness}). Then there exists a function $\delta\colon\R_{>0}\to\R_{>0}$ such that for every $k\in\N$, $(k\text{-SAT},2^n)\leq_{\delta} (\textrm{MM}_{\F_q},n^\gamma)$. Let us fix $\eps_0=(\gamma-\gamma')/(2\gamma)\in(0,1)$ and $\delta_0=\delta(\eps_0)\in(0,1)$, where $\delta$ is the function from the definition of SETH-hardness. For every $k$, let $R_k$ be the fine-grained reduction from $k$-SAT to $\textrm{MM}_{\F_q}$ guaranteed by the $n^\gamma$-SETH-hardness of $\textrm{MM}_{\F_q}$. An $n^{\gamma(1-\eps_0)}$-time algorithm for $\textrm{MM}_{\F_q}$ would imply (via $R_k$) a $2^{n(1-\delta_0)}$-time algorithm for $k$-SAT.
    In particular, the reduction $R_k$ runs in time $O(2^{n(1-\delta_0)})$.
    
    For every $k$, let us consider the following algorithm ${\cal A}_k$ equipped with an $\NP$ oracle. On input~$1^N$, if $N$ is not a power of two, the algorithm halts. Otherwise, the algorithm sets $n=\log_2{N}$ and $\lambda=\delta_0/4$. Let $c=c(k,\lambda)$ be the constant from the Sparsification Lemma (\cref{sparsification_lemma}). The algorithm ${\cal A}_k$ enumerates all $k$-SAT formulas with $n$ variables and at most $cn$ clauses, applies the fine-grained reduction~$R_k$ to each of them, solves all queried instances of $\textrm{MM}_{\F_q}$ by the trivial cubic-time algorithm, and this way adaptively obtains all MM instances produced by~$R_k$. Let us denote these instances of MM by $(A_1,B_1),\ldots,(A_T,B_T)$. For each instance $(A_i, B_i)$ where $A_i,B_i\in\F_q^{M\times M}$, if $M>2^{\delta_0n/12}$ the algorithm checks if it satisfies
    \[
    \cR_{M^{r(\gamma')}}^{\F_q}(A_i),\cR_{M^{r(\gamma')}}^{\F_q}(B_i)\geq M^{s(\gamma')} \,.
    \]
    Since the problem of checking rigidity is trivially in $\coNP$, ${\cal A}_k$ checks this using the $\NP$ oracle. If ${\cal A}_k$ finds a rigid matrix, it outputs the first such found matrix, otherwise it outputs nothing.

    Now we bound the running time of the algorithms ${\cal A}_k$. Enumeration of all $k$-SAT instances with $n$ variables and $cn$ clauses takes time $\binom{O(n^k)}{\leq cn}=2^{O(n\log{n})}$. The running time of the Sparsification Lemma is $2^{\lambda n}=2^{\delta_0 n/4}$. The running time of each execution of the fine-grained reduction is $O(2^{n(1-\delta_0)})$, and the time needed to perform each matrix multiplication is at most $O(2^{3n(1-\delta_0)})$. This leads to the total running time of $2^{O(n\log{n})}=N^{O(\log\log{N})}$.

    If one of the constructed algorithms ${\cal A}_k$ outputs matrices on inputs $1^N$ infinitely often, then we are done, as we have a $\QP^\NP$ algorithm outputting rigid matrices infinitely often. Indeed, the algorithms ${\cal A}_k$ output only rigid matrices, and the dimension of each such matrix is $M>2^{\delta_0 n / 12}=N^{\Theta(1)}$.
    
    In the following, we assume that for each $k$, there exists $N_k$ such that the algorithm ${\cal A}_k$ does not output matrices on inputs $1^N$ for all $N\geq N_k$. Equivalently, there exists $n_k$ such that all $k$-SAT formulas with $n\geq n_k$ variables and $cn$ clauses, when fed to the reduction~$R_k$, result in instances $(A_i,B_i)$ of $\textrm{MM}_{\F_q}$ with $A_i,B_i\in\F_q^{M\times M}$ either of size $M\leq2^{\delta_0 n/12}$ or satisfying 
    \begin{align}\label{eq:rig}
        \cR_{M^{r(\gamma')}}^{\F_q}(A_i)< M^{s(\gamma')}\text{\;\; and \;\;}
        \cR_{M^{r(\gamma')}}^{\F_q}(B_i)< M^{s(\gamma')}\; .
    \end{align}
    We will use this to design a non-deterministic algorithm that solves $k$-TAUT in time $O(2^{n(1-\delta_0)})$ for every $k$, which refutes NSETH.

    Given a $k$-DNF formula $\varphi$, we first apply the Sparsification Lemma to (the negation of) $\varphi$ with $\lambda=\delta_0/4$. This gives us $2^{\lambda n}$ $k$-CNF formulas $\varphi_i$ such that $\varphi$ is a tautology if and only each $\varphi_i$ is unsatisfiable. Therefore, solving $k$-TAUT for the negation of each $\varphi_i$ is sufficient for solving the original instance $\varphi$. Now we apply the reduction~$R_k$ to each $\varphi_i$. Recall that the fine-grained reduction runs in time $2^{(1-\delta_0)n}$, and, in particular, creates at most $2^{n(1-\delta_0)}$ instances of MM, each of size at most $2^{n(1-\delta_0)}$. Therefore, the total number $T$ of instances of $\textrm{MM}_{\F_q}$ we obtain is at most $T\leq2^{\lambda n}\cdot 2^{n(1-\delta_0)}=2^{(1-3\delta_0/4)n}$. We solve each instance of size $M\leq 2^{\delta n /12}$ by a trivial cubic-time algorithm. The overall running time to solve all such small instances is at most $T\cdot 2^{3\delta_0 n /12} \leq 2^{n(1-\delta_0/2)}$.
    
    Now we use the assumption that the algorithm ${\cal A}_k$ does not find rigid matrices for all but finitely many inputs. Therefore, for every $n\geq n_k$, every instance $(A_i, B_i)$, $A_i,B_i\in\F^{M\times M}$ where $M> 2^{\delta_0 n/12}$ satisfies \cref{eq:rig}.
    We apply \cref{lem:nondet} with $\alpha=(\gamma-\gamma')$ to compute $A_i\cdot B_i$ in non-deterministic time $M^{\gamma'+\alpha}=M^{(\gamma+\gamma')/2}=M^{\gamma(1-\eps_0)}$. By the guarantee of the SETH-hardness reduction~$R_k$, we have that a non-deterministic algorithm solving $\textrm{MM}_{\F_q}$ in time $M^{\gamma(1-\eps_0)}$ implies a non-deterministic algorithm for $k$-TAUT with running time $2^{n(1-\delta_0)}$. This refutes NSETH, and finishes the proof of the theorem.
\end{proof}

\subsection{Comparison with Known Rigidity Results} 
\label{sec:rigidity-comparison}

While the rigidity parameters that come out of $n^{\gamma}$-SETH-hardness of matrix multiplication via \cref{thm:non-seth-hardness} are not strong enough for Valiant's program for proving circuit lower bounds, they would improve on all currently known constructions of rigid matrices. The main problem in this area is to find an \emph{explicit} matrix (i.e., a matrix from a complexity class where we cannot prove circuit lower bounds: $\P,\NP$, or even $\cc{E}^\NP$) of high rigidity for linear rank $r=\Theta(n)$.

The best known rigidity lower bound for matrices constructible in polynomial time (i.e., from the class~$\P$) is $\cR_r^{\F}(A)\geq \Omega\left(\frac{n^2}{r}\cdot\log{\frac{n}{r}}\right)$~\cite{friedman1993note,pudlak1994some,shokrollahi1997remark}, and it falls short of achieving the rigidity parameters needed for Valiant's program (that requires rigidity $\Omega(n^{1+\eps})$ for $r=\Theta(n)$ for a constant~$\eps>0$). Goldreich and Tal~\cite{goldreich2016matrix} give a matrix $A\in\F^{n\times n}$ constructible in time $2^{O(n)}$ (i.e., in the class~$\cc{E}$) of rigidity $\cR_r^\F(A)\geq\Omega \left(\frac{n^3}{r^2\log(n)}\right)$ for any $r\geq\sqrt{n}$, which improves on the previous constructions for $r=o\left(\frac{n}{\log(n)\log\log(n)}\right)$. Alman and Chen, and Bhangale et al.~\cite{alman2022efficient,bhangale2020rigid} give a matrix $A\in\F_2^{n\times n}$ achieving very high rigidity $\cR_r^{\F_2}(A)\geq\Omega(n^2)$ for sub-polynomial rank $r=2^{\eps\log(n)/\log\log(n)}$ in $\P^\NP$.

We now compare the rigidity bounds implied by \cref{thm:non-seth-hardness} from $n^\gamma$-SETH-hardness for various values of $\gamma>2$ to the known lower bounds on rigidity. For this, we use the bounds from \cref{eq:rbound,eq:sbound} and the values of $\omega$ and $\dualMMexp$ from \cref{eq:MMexp-UB,eq:dualMMexp-LB}. The rigidity bound implied by \cref{thm:non-seth-hardness} from $n^\gamma$-SETH-hardness for $\gamma  \geq 2.37$ would give us a matrix~$A$ of rigidity $\cR_r^{\F}(A)\geq n^{1.68}$ for rank as high as $r=n^{0.997}$. The rigidity bound obtained from $\gamma \geq 2.24$  would improve upon the best known construction in $\cc{E}$~\cite{goldreich2016matrix} for every $r\geq n^{0.751}$. The rigidity bound obtained from any $\gamma \geq 2.17$ would already improve upon the best known constructions in $\P$~\cite{friedman1993note,pudlak1994some,shokrollahi1997remark} for every $r\geq n^{0.63}$.  Finally, the bound obtained from any $\gamma \geq 2$ would give us rigidity for higher rank than the known $\P^{\NP}$ constructions~\cite{alman2022efficient,bhangale2020rigid}.

\section{Reductions Between Variants of MMV}
\label{sec:reductions}

In this section, we describe variants of MMV and study their relative difficulty compared to (standard) MMV and each other under deterministic reductions using $O(n^2)$ ring operations. (For clarity and conciseness, in this section we assume that basic arithmetic operations over arbitrary rings take constant time, and so refer to these reductions simply as $O(n^2)$-time deterministic operations.)
See \cref{fig:web-of-reductions} for a summary of these variants and reductions between them.

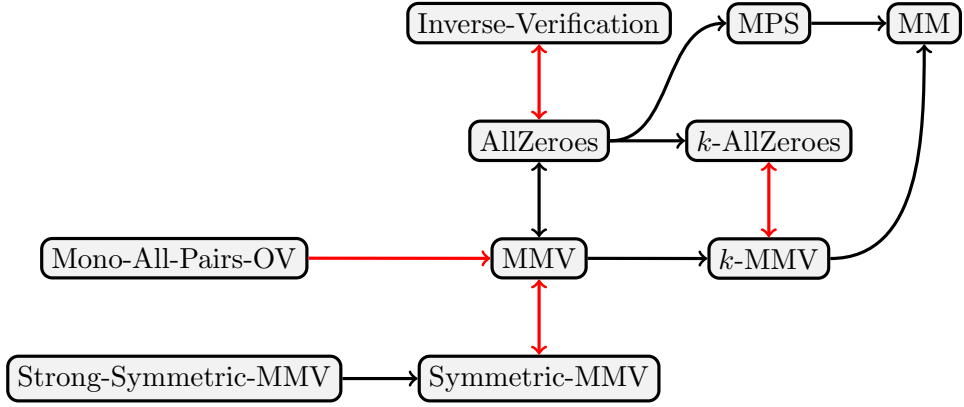
\begin{figure}[t]
    \centering
    \begin{tikzpicture}[
        problem/.style={rectangle, rounded corners, draw=black, fill=black!5, very thick, minimum size=5mm},
    ]
        \node[problem] (mmv) {\MMV};
        \node[problem] (all zeroes) [above=of mmv]{\AZ};
        \node[problem] (inverse verification) [above=of all zeroes]{\IV};
        \node[problem] (symmetric mmv) [below=of mmv] {\SymMMV};
        \node[problem] (strongly symmetric mmv) [left=of symmetric mmv] {\StrongSymMMV};
        \node[problem] (k all zeroes) [right=of all zeroes] {\kAZ};
        \node[problem] (pairwise orthogonality) [above=of strongly symmetric mmv] {Mono-All-Pairs-OV};
        \node[problem] (kmmv) [below=of k all zeroes] {\kMMV};
        \node[problem] (matrix product sparsity) [above=of k all zeroes] {\MPS};
        \node[problem] (mm) [right=of matrix product sparsity] {\MM};

        \draw[->, very thick] (strongly symmetric mmv.east) to (symmetric mmv.west);
        \draw[<->, very thick, red] (symmetric mmv.north) to (mmv.south);
        \draw[<->, very thick, red] (all zeroes.north) to (inverse verification.south);
        \draw[<->, very thick] (all zeroes.south) to (mmv.north);
        \draw[->, very thick] (all zeroes.east) to[out=0, in=180] (matrix product sparsity.west);
        \draw[->, very thick] (mmv.east) to[out=0, in=180] (kmmv.west);
        \draw[<->, very thick, red] (kmmv.north) to (k all zeroes.south);
        \draw[->, very thick] (matrix product sparsity.east) to (mm.west);
        \draw[->, very thick] (kmmv.east) to[out=0, in=270] (mm.south);
        \draw[->, very thick] (all zeroes.east) to[out=0, in=180] (k all zeroes.west);
        \draw[->, very thick, red] (pairwise orthogonality.east) to[out=0, in=180] (mmv.west);
    \end{tikzpicture}
    \caption{A diagram of reductions among $\MMV$ and related problems on $n \times n$ matrices. Arrows represent deterministic $O(n^2)$-time reductions (and double-headed arrows denote equivalences under such reductions).
    {\color{red} Red} arrows indicate new (non-trivial) reductions.
    }
    \label{fig:web-of-reductions}
\end{figure}

\subsection{Problems Equivalent to MMV}

In this section we present two problems (special cases of MMV) that are in fact equivalent to (standard) MMV under deterministic $O(n^2)$-time reductions.
These problems are: (1) the Inverse Verification Problem, and (2) the Symmetric MMV problem.

\subsubsection{Inverse Verification}
\label{sec:inverse-verification}

It is known that matrix multiplication is $O(n^2)$-time reducible to \emph{matrix inversion}~\cite[Proposition 16.6]{algebraic-complexity-theory-1997}, and in fact that the problems are equivalent under $O(n^2)$-time reductions.
Here we use a very similar reduction to relate the verification analogs of these problems.
We first define the Inverse Verification Problem.

\begin{definition}[$\IV_R$]
The Inverse Verification Problem over a ring $R$ ($\IV_R$) is defined as follows.
Given matrices $A, B \in R^{n \times n}$ for some $n \in \Z^+$ as input, decide whether $B = A^{-1}$.
\end{definition}

We next give the reduction.

\begin{theorem} \label{thm:all_zeroes_to_inverse_verif}
For any ring $R$, there is an $O(n^2)$-time reduction from $\AZ_R$ on $n \times n$ matrices to $\IV_R$ on $3n \times 3n$ matrices.
\end{theorem}

\begin{proof}
Let $A, B \in R^{n \times n}$ be the input instance of $\AZ_R$, and define the block matrices
\[
    A' = \begin{bmatrix}
        I & A & 0 \\
        0 & I & B \\
        0 & 0 & I
    \end{bmatrix}, \ \ \ \ 
    B' = \begin{bmatrix}
        I & -A & 0 \\
        0 & I & -B \\
        0 & 0 & I
    \end{bmatrix} \ \text{.}
\]
It is straightforward to check that 
\[
    A'B' =
    \begin{bmatrix}
        I & 0 & -AB \\
        0 & I & 0 \\
        0 & 0 & I
    \end{bmatrix} \ \text{,}
\]
and so $A'B' = I_{3n}$ if and only if $AB = 0$.
\end{proof}

\subsubsection{Symmetric MMV}
$\IV$ is not the only variant (special case) of $\MMV$ that is equivalent to $\MMV$.
Suppose that the input matrices $A$ and $B$ in an $\MMV$ instance are symmetric.
One might hope this would make $\MMV$ easier, however we next show that this is not the case.

\begin{definition}[$\SymMMV_R$]
The Symmetric Matrix Multiplication Verification Problem over a ring $R$ ($\SymMMV_R$) is defined as follows.
Given symmetric matrices $A, B \in R^{n \times n}$ and a matrix $C \in R^{n \times n}$ for some $n \in \Z^+$ as input, decide whether $AB = C$.
\end{definition}
Note that this definition does not require that $C$ be a symmetric matrix.

\begin{theorem} \label{thm:mmv_to_symMmv}
For any ring $R$, there is an $O(n^2)$-time reduction from $\MMV_R$ on $n \times n$ matrices to $\SymMMV_R$ on $3n \times 3n$ matrices. 
\end{theorem}
\begin{proof}

Let $A, B, C \in R^{n \times n}$ be the input instance of $\MMV_R$, and define matrices $A', B', C'$ as follows:
\[
    A' = \begin{bmatrix}
        I & 0 & A \\
        0 & I & -A \\
        A^T & -A^T & I
    \end{bmatrix}, \quad
    B' = \begin{bmatrix}
        I & 0 & B^T \\
        0 & I & B^T \\
        B & B & I
    \end{bmatrix}, \quad
    C' = \begin{bmatrix}
        I + C & C & B^T + A \\
        -C & I - C & B^T - A \\
        A^T + B & -A^T + B & I
    \end{bmatrix} \ \text{.}
\]
It is easy to check that $A'$ and $B'$ are both symmetric, and, because adding and taking the tranpose of $n \times n$ matrices takes $O(n^2)$ time, $A', B', C'$ are all computable in $O(n^2)$ time.

Moreover, it is easy to check that
\[
A' B' =     \begin{bmatrix}
        I + AB & AB & B^T + A \\
        -AB & I - AB & B^T - A \\
        A^T + B & -A^T + B & I
    \end{bmatrix} \ \text{,}
\]
and hence that $A' B' = C'$ if and only if $AB = C$.
\end{proof}

\subsection{Problems no Harder than MMV}
\label{sec:MMV-easy}
Next, in this section, we present two problems that are $O(n^2)$-time reducible to MMV, but that are not obviously equivalent to MMV and are not obviously solvable in less time than MMV. The problems are: (1) the Monochromatic All-Pairs Orthogonal Vectors Problem, and (2) the Strong Symmetric MMV Problem. We propose these problems (as special cases of MMV) as good targets for making algorithmic progress on MMV.

\subsubsection{Monochromatic All Pairs Orthogonal Vectors}

\begin{definition}[$\MCAPO_R$]
The Monochromatic All Pairs Orthogonal Vectors Problem over a ring $R$ ($\MCAPO_R$) is defined as follows. Given vectors $\vec{a}_1, \ldots, \vec{a}_n \in R^n$ for some $n \in \Z^+$, decide whether for all $i, j \in [n], i \neq j$, $\iprod{\vec{a}_i, \vec{a}_j} = 0$.
\end{definition}

We next show that $\MCAPO$ reduces to $\MMV_{\Z}$.

\begin{lemma}\label{lem:monochrom-all-pairs-orth-reduces-to-mmv}
There is a $O(n^2)$-time reduction from $\MCAPO_R$ on $n$ vectors of length $n$ to $\MMV_{R}$ on $n \times n$ matrices.
\end{lemma}

\begin{proof}
Let $\vec{a}_1, \ldots, \vec{a}_n \in R^n$ be the input instance of $\MCAPO$, 
and define $A := (\vec{a}_1,\ldots,\vec{a}_n) \in R^{n \times n}$. Additionally, compute
    \[
        C = \diag(\iprod{\vec{a}_1, \vec{a}_1}, \iprod{\vec{a}_2, \vec{a}_2}, \ldots, \iprod{\vec{a}_n, \vec{a}_n}) =
        \begin{bmatrix}
           \iprod{\vec{a}_1, \vec{a}_1} & 0 & \cdots & 0 \\
            0 & \iprod{\vec{a}_2, \vec{a}_2} & \cdots & 0 \\
            \vdots & \vdots & \ddots & 0 \\
            0 & 0 & \cdots & \iprod{\vec{a}_n, \vec{a}_n}
        \end{bmatrix} \ \text{,}
    \]
and output the $\MMV_{R}$ instance $A^T, A, C$.
It is straightforward to check that $A^T A = C$ if and only if for all $i, j \in [n], i \neq j$, $\iprod{\vec{a}_i, \vec{a}_j} = 0$. Moreover, computing each of $A$, $A^T$, and $C$ requires $O(n^2)$ time, as needed.
\end{proof}

We note that, interestingly, while it is relatively straightforward to reduce $\MCAPO_R$ to $\MMV_R$ in $O(n^2)$ time, it is not clear that such a reduction in the other direction exists. In contrast, the bichromatic analogue of $\MCAPO$ (in which given $\vec{a}_1, \ldots, \vec{a}_n, \vec{b}_1, \ldots, \vec{b}_n$ the goal is to decide whether $\iprod{\vec{a}_i, \vec{b}_j} = 0$ for all $1 \leq i, j \leq n$) is equivalent to $\AZ$.

\subsubsection{Strong Symmetric MMV}
For our second problem that is no harder than $\MMV$, we consider the special case of $\SymMMV$ when $C$ must also be symmetric.
We call this problem $\StrongSymMMV$ and define it as follows.
\begin{definition}[$\StrongSymMMV_R$]
The Strong Symmetric Matrix Multiplication Verification Problem over a ring $R$ ($\StrongSymMMV_R$) is defined as follows. Given symmetric matrices $A, B, C \in R^{n \times n}$ for some $n \in \Z^+$ as input, decide whether $AB = C$.
\end{definition}

We note that the reduction in \cref{thm:mmv_to_symMmv} does \emph{not} work as a reduction to $\StrongSymMMV$. On the other hand, it is not clear how to solve $\StrongSymMMV$ more efficiently than $\MMV$. However, finding faster algorithms for $\StrongSymMMV$ seems like a good starting point for finding faster algorithms for general $\MMV$. 

We make several structural observations to this end.
First, we note that the product of two symmetric matrices is itself symmetric if and only if the two matrices commute.
That is, for symmetric matrices $A$ and $B$, $AB=(AB)^T$ if and only if $AB=BA$.
Second, since symmetric matrices (over a field $\F$) are always diagonalizable, we know that two diagonalizable matrices commute if and only if they are simultaneously diagonalizable (see \cite[Theorem 1.3.21]{matrix-analysis-horn-johnson}).%
\footnote{We say that matrices $A, B \in \F^{n \times n}$ are \textit{simultaneously diagonalizable} if there exists some matrix $P \in \F^{n \times n}$ such that $P A P^{-1}$ and $P B P^{-1}$ are both diagonal matrices.}

While it seems like it may be possible to exploit these properties to find a faster algorithm, there are some caveats.
The added restriction that $C$ is symmetric still does not guarantee that $AB$ is itself symmetric, meaning the above properties ($A$ and $B$ commuting and being simultaneously diagonalizable) are not guaranteed, so we cannot assume them to be true.
That is, if $AB$ is symmetric and therefore $A$ and $B$ commute and are simultaneously diagonalizable, it does not guarantee that $AB=C$.
So, although these observations about symmetric matrices $A$ and $B$ may be useful, they do not obviously lead to a faster algorithm for $\StrongSymMMV$.

\subsection{Problems at Least as Hard as MMV but no Harder than MM}
Finally, we present two problems to which MMV is deterministically $O(n^2)$-time reducible, and that are deterministically $O(n^2)$-time reducible to MM. These problems are: (1) the $k$-MMV problem, and (2) the Matrix Product Sparsity Problem (MPS). We also present the $k$-$\AZ$ problem, and show that it is equivalent to $k$-$\MMV$.

\subsubsection{Generalized \texorpdfstring{$\AZ$}{AllZeroes} Equivalence}

So far, we have only discussed verifying the product of \emph{two} matrices. But, it is also natural to ask about the complexity of verifying the product of $k > 2$ matrices.
Accordingly, we define the following problems.

\begin{definition}[$\kMMV_R$]
Let $k \geq 2$ be a fixed integer.
The $k$-Matrix Multiplication Verification Problem over a ring $R$ ($\kMMV_R$) is defined as follows. Given matrices $A_1, \ldots, A_k, C \in R^{n \times n}$ for some $n \in \Z^+$ as input, decide whether $\prod_{i=1}^k A_i = C$.
\end{definition}

\begin{definition}[$\kAZ_R$]
Let $k \geq 2$ be a fixed integer.
The $k$-All Zeroes Problem over a ring $R$ ($\kAZ_R$) is defined as follows. Given matrices $A_1, \ldots, A_k \in R^{n \times n}$ for some $n \in \Z^+$ as input, decide whether $\prod_{i=1}^k A_i = 0$.
\end{definition}

We note that $2$-$\MMV_R$ and $2$-$\AZ_R$ are ``standard'' $\MMV_R$ and $\AZ_R$, respectively.

It turns out that many algorithms (including ours) for $\MMV$ also work for $k$-$\MMV$, and in fact we can generalize certain results about the complexity of $\MMV$.
Additionally, it is easy to reduce $\MMV$ to $k$-$\MMV$ for any $k \geq 2$, but the reduction in the opposite direction is unclear.
We first show that $k$-$\MMV$ and $k$-$\AZ$ are equivalent.

\begin{theorem} \label{thm:k_MMV_to_k_AllZeroes}
    Let $k$ be some fixed positive integer, and let $R$ be a ring.
    There is a $O(n^2)$-time reduction from $k$-$\MMV_R$ on $n \times n$ matrices to $k$-$\AZ_R$ on $2n \times 2n$ matrices. 
\end{theorem}
\begin{proof}
    Let $A_1, \ldots, A_n, C \in R^{n \times n}$ be the input instance of $k$-$\MMV_R$. We output
    \begin{align*}
    A_1' &:=  \begin{bmatrix}
        A_1 & I \\
        0 & 0
    \end{bmatrix} \ \text{,} \\
    A_i' &:=  \begin{bmatrix}
        A_i & I \\
        0 & 0
    \end{bmatrix} \text{ for $1 < i < k$, and} \\
    A_k' &:= \begin{bmatrix}
        A_k & 0 \\
        -C & 0
    \end{bmatrix}
    \end{align*}
    as our $k$-$\AZ_R$ instance.
    
    Note that 
    \[
    \prod_{i=1}^k A_i' =
    \begin{bmatrix}
        A_1 & I \\
        0 & 0
    \end{bmatrix}
    \cdot
    \prod_{i=2}^{k-1}
    \begin{bmatrix}
        A_i & 0 \\
        0 & I
    \end{bmatrix}
    \cdot
    \begin{bmatrix}
        A_k & 0 \\
        -C & 0
    \end{bmatrix}
    =
    \begin{bmatrix}
        \prod_{i=1}^k A_i - C & 0 \\
        0 & 0
    \end{bmatrix} \ \text{,}
    \]
    and that this product equals the all-zeroes matrix if and only if $\prod_{i=1}^k A_i = C$.
\end{proof}

We also note that the reduction in the other direction (from $k$-$\AZ_R$ to $k$-$\MMV_R$) follows immediately by setting $C = 0$.

\subsubsection{Relation to Orthogonal Vectors and \MPSlong{}}

We noted previously how both the monochromatic and bichromatic versions of All-Pairs-OV reduce to $\MMV$, and how the bichromatic version is equivalent to $\AZ$.
However, essentially all variations of Orthogonal Vectors can be generalized by the (monochromatic or bichromatic) \emph{counting version} of the Orthogonal Vectors Problem ($\SOV$), in which the goal is to decide whether the number of orthogonal pairs is less than or equal to an input parameter $t$.
It is natural then to study how $\SOV$ is related to $\MMV$.
To do so, we introduce an equivalent matrix formulation of $\SOV$ that we call the \emph{Matrix Product Sparsity Problem}, the goal of which is to decide whether $\norm{AB}_0 \leq t$ for some input matrices $A, B$ and number $t$.

\begin{definition}[$\MPS_R$]
The Matrix Product Sparsity Problem over a ring $R$ ($\MPS_R$) is defined as follows. Given matrices $A, B \in R^{n \times n}$ for some $n \in \Z^+$ and $t \in \Z^+$ as input, decide whether $\norm{AB}_0 \leq t$.
\end{definition}

It immediately follows that $\MPS$ is intermediate in difficulty to $\MMV$ and $\MM$.
\begin{lemma}
Let $R$ be a ring. There is an $O(n^2)$-time reduction from $\MMV_R$ on $n \times n$ matrices to $\MPS_R$ on $n \times n$ matrices, and from $\MPS_R$ on $n \times n$ matrices to $\MM_R$ on $n \times n$ matrices.
\end{lemma}

\begin{proof}
The reduction from $\MMV_R$ goes through $\AZ_R$, to which $\MMV_R$ is equivalent. Indeed, $\AZ_R$ is the special case of $\MPS_R$ when $t = 0$.
For the reduction from $\MPS_R$ to $\MM_R$, given an instance $A, B, t$ of $\MPS_R$ as input, compute $AB$ and then count the number of non-zero entries in $AB$ in $O(n^2)$ time. 
\end{proof}

\bibliographystyle{alpha}
\bibliography{mmv}

\newcommand{\etalchar}[1]{$^{#1}$}
\begin{thebibliography}{BKM{\etalchar{+}}23}

\bibitem[ABB{\etalchar{+}}23]{aggarwal2023lattice}
Divesh Aggarwal, Huck Bennett, Zvika Brakerski, Alexander Golovnev, Rajendra
  Kumar, Zeyong Li, Spencer Peters, Noah Stephens-Davidowitz, and Vinod
  Vaikuntanathan.
\newblock Lattice problems beyond polynomial time.
\newblock In {\em STOC}, 2023.

\bibitem[ABFK24]{abboud2024time}
Amir Abboud, Karl Bringmann, Nick Fischer, and Marvin K{\"u}nnemann.
\newblock The time complexity of fully sparse matrix multiplication.
\newblock In {\em SODA}, 2024.

\bibitem[AC19]{alman2022efficient}
Josh Alman and Lijie Chen.
\newblock Efficient construction of rigid matrices using an {NP} oracle.
\newblock In {\em FOCS}, 2019.

\bibitem[AGHP90]{alon-1990}
Noga Alon, Oded Goldreich, Johan H{\aa}stad, and Ren{\'e} Peralta.
\newblock Simple construction of almost $k$-wise independent random variables.
\newblock In {\em FOCS}, 1990.

\bibitem[AV21]{alman2021refined}
Joah Alman and Virginia {Vassilevska Williams}.
\newblock A refined laser method and faster matrix multiplication.
\newblock In {\em SODA}, 2021.

\bibitem[BCS97]{algebraic-complexity-theory-1997}
Peter Bürgisser, Michael Clausen, and Amin Shokrollahi.
\newblock {\em Algebraic Complexity Theory}, volume 315.
\newblock Springer, 1997.

\bibitem[BGK{\etalchar{+}}23]{belova2023polynomial}
Tatiana Belova, Alexander Golovnev, Alexander~S. Kulikov, Ivan Mihajlin, and
  Denil Sharipov.
\newblock Polynomial formulations as a barrier for reduction-based hardness
  proofs.
\newblock In {\em SODA}, 2023.

\bibitem[BHPT20]{bhangale2020rigid}
Amey Bhangale, Prahladh Harsha, Orr Paradise, and Avishay Tal.
\newblock Rigid matrices from rectangular {PCPs} or: Hard claims have complex
  proofs.
\newblock In {\em FOCS}, 2020.

\bibitem[BKM{\etalchar{+}}23]{belova2023computations}
Tatiana Belova, Alexander~S. Kulikov, Ivan Mihajlin, Olga Ratseeva, Grigory
  Reznikov, and Denil Sharipov.
\newblock Computations with polynomial evaluation oracle: ruling out
  superlinear {SETH}-based lower bounds.
\newblock {\em arXiv:2307.11444}, 2023.

\bibitem[BS04]{buhrman2004quantum}
Harry Buhrman and Robert Spalek.
\newblock Quantum verification of matrix products.
\newblock {\em arXiv:quant-ph/0409035}, 2004.

\bibitem[CGI{\etalchar{+}}16]{carmosino2016nondeterministic}
Marco~L. Carmosino, Jiawei Gao, Russell Impagliazzo, Ivan Mihajlin, Ramamohan
  Paturi, and Stefan Schneider.
\newblock Nondeterministic extensions of the strong exponential time hypothesis
  and consequences for non-reducibility.
\newblock In {\em ITCS}, 2016.

\bibitem[CKSU05]{cohn2005group}
Henry Cohn, Robert Kleinberg, Bal{\'a}zs Szegedy, and Christopher Umans.
\newblock Group-theoretic algorithms for matrix multiplication.
\newblock In {\em FOCS}, 2005.

\bibitem[DWZ23]{duan2022faster}
Ran Duan, Hongxun Wu, and Renfei Zhou.
\newblock Faster matrix multiplication via asymmetric hashing.
\newblock In {\em FOCS}, 2023.

\bibitem[Fre79]{freivalds1979fast}
R{\=u}si{\c{n}}{\v{s}} Freivalds.
\newblock Fast probabilistic algorithms.
\newblock In {\em MFCS}, 1979.

\bibitem[Fri93]{friedman1993note}
Joel Friedman.
\newblock A note on matrix rigidity.
\newblock {\em Combinatorica}, 13:235--239, 1993.

\bibitem[GLL{\etalchar{+}}17]{journals/algorithmica/GasieniecLLPT17}
Leszek Gasieniec, Christos Levcopoulos, Andrzej Lingas, Rasmus Pagh, and
  Takeshi Tokuyama.
\newblock Efficiently correcting matrix products.
\newblock {\em Algorithmica}, 79(2):428--443, 2017.

\bibitem[GRS22]{grs-essential-coding-theory}
Venkatesan Guruswami, Atri Rudra, and Madhu Sudan.
\newblock {\em Essential Coding Theory}.
\newblock Online Draft, 2022.

\bibitem[GT16]{goldreich2016matrix}
Oded Goldreich and Avishay Tal.
\newblock Matrix rigidity of random {Toeplitz} matrices.
\newblock In {\em STOC}, 2016.

\bibitem[Hal]{hall-coding}
Jonathan~I. Hall.
\newblock Notes on coding theory.
\newblock Available at
  \url{https://users.math.msu.edu/users/halljo/classes/codenotes/GRS.pdf}.

\bibitem[HJ12]{matrix-analysis-horn-johnson}
Roger~A. Horn and Charles~R. Johnson.
\newblock {\em Matrix Analysis}.
\newblock Cambridge University Press, 2012.

\bibitem[HTW23]{hon2023verifying}
Wing-Kai Hon, Meng-Tsung Tsai, and Hung-Lung Wang.
\newblock Verifying the product of generalized boolean matrix multiplication
  and its applications to detect small subgraphs.
\newblock In {\em WADS}, 2023.

\bibitem[IP99]{IP99}
Russell Impagliazzo and Ramamohan Paturi.
\newblock The complexity of {$k$-SAT}.
\newblock In {\em CCC}, 1999.

\bibitem[IPZ98]{impagliazzo2001problems}
Russell Impagliazzo, Ramamohan Paturi, and Francis Zane.
\newblock Which problems have strongly exponential complexity?
\newblock In {\em FOCS}, 1998.

\bibitem[IS09]{journals/ipl/IwenS09}
Mark~A. Iwen and Craig~V. Spencer.
\newblock A note on compressed sensing and the complexity of matrix
  multiplication.
\newblock {\em Information Processing Letters}, 109(10):468--471, 2009.

\bibitem[Iwe23]{Iwen-Personal-23}
Mark~A. Iwen, 2023.
\newblock Personal Communication.

\bibitem[JMV15]{jahanjou2015local}
Hamid Jahanjou, Eric Miles, and Emanuele Viola.
\newblock Local reductions.
\newblock In {\em ICALP}, 2015.

\bibitem[KS93]{kimbrel1993probabilistic}
Tracy Kimbrel and Rakesh~Kumar Sinha.
\newblock A probabilistic algorithm for verifying matrix products using
  {$O(n^2)$} time and {$\log_2 n + O(1)$} random bits.
\newblock {\em Information Processing Letters}, 45(2):107--110, 1993.

\bibitem[K{\"u}n18]{kunnemann2018nondeterministic}
Marvin K{\"u}nnemann.
\newblock On nondeterministic derandomization of {Freivalds'} algorithm:
  Consequences, avenues and algorithmic progress.
\newblock In {\em ESA}, 2018.

\bibitem[KW14]{korec2014deterministic}
Ivan Korec and Ji{\v{r}}{\'\i} Wiedermann.
\newblock Deterministic verification of integer matrix multiplication in
  quadratic time.
\newblock In {\em SOFSEM}, 2014.

\bibitem[{Le }12]{le2012faster}
Fran{\c{c}}ois {Le Gall}.
\newblock Faster algorithms for rectangular matrix multiplication.
\newblock In {\em FOCS}, 2012.

\bibitem[{Le }24]{gall2023faster}
Fran{\c{c}}ois {Le Gall}.
\newblock Faster rectangular matrix multiplication by combination loss
  analysis.
\newblock In {\em SODA}, 2024.

\bibitem[LR83]{lotti1983asymptotic}
Grazia Lotti and Francesco Romani.
\newblock On the asymptotic complexity of rectangular matrix multiplication.
\newblock {\em Theoretical Computer Science}, 23(2):171--185, 1983.

\bibitem[LU18]{gall2018improved}
Fran{\c{c}}ois {Le Gall} and Florent Urrutia.
\newblock Improved rectangular matrix multiplication using powers of the
  {Coppersmith-Winograd} tensor.
\newblock In {\em SODA}, 2018.

\bibitem[PR94]{pudlak1994some}
Pavel Pudl{\'a}k and Vojtech R{\"o}dl.
\newblock Some combinatorial-algebraic problems from complexity theory.
\newblock {\em Discrete Mathematics}, 1(136):253--279, 1994.

\bibitem[PS72]{polya1972problems}
George P{\'o}lya and Gabor Szeg{\"o}.
\newblock {\em Problems and Theorems in Analysis: Series, integral calculus,
  theory of functions}.
\newblock Springer, 1972.

\bibitem[RL89]{roth1989mds}
Ron~M. Roth and Abraham Lempel.
\newblock {On MDS codes via Cauchy matrices}.
\newblock {\em IEEE Transactions on Information Theory}, 35(6):1314--1319,
  1989.

\bibitem[Sho90]{shoup1990new}
Victor Shoup.
\newblock New algorithms for finding irreducible polynomials over finite
  fields.
\newblock {\em Mathematics of computation}, 54(189):435--447, 1990.

\bibitem[Sho09]{shoup2009computational}
Victor Shoup.
\newblock {\em A computational introduction to number theory and algebra}.
\newblock Cambridge University Press, 2009.

\bibitem[SSS97]{shokrollahi1997remark}
Mohammad~Amin Shokrollahi, Daniel~A. Spielman, and Volker Stemann.
\newblock A remark on matrix rigidity.
\newblock {\em Information Processing Letters}, 64(6):283--285, 1997.

\bibitem[Val77]{V77}
Leslie~G. Valiant.
\newblock Graph-theoretic arguments in low-level complexity.
\newblock In {\em MFCS}, 1977.

\bibitem[{Vas}15]{V15}
Virginia {Vassilevska Williams}.
\newblock Hardness of easy problems: Basing hardness on popular conjectures
  such as the strong exponential time hypothesis (invited talk).
\newblock In {\em IPEC}, 2015.

\bibitem[{Vas}18]{V18}
Virginia {Vassilevska Williams}.
\newblock On some fine-grained questions in algorithms and complexity.
\newblock In {\em ICM}, 2018.

\bibitem[vzGG13]{modernalgebra}
Joachim von~zur Gathen and Jürgen Gerhard.
\newblock {\em Modern Computer Algebra}.
\newblock Cambridge University Press, 3 edition, 2013.

\bibitem[WW23]{wu2023correcting}
Yu-Lun Wu and Hung-Lung Wang.
\newblock Correcting matrix products over the ring of integers.
\newblock {\em arxiv:2307.12513}, 2023.

\bibitem[WXXZ24]{williams2024new}
Virginia~Vassilevska Williams, Yinzhan Xu, Zixuan Xu, and Renfei Zhou.
\newblock New bounds for matrix multiplication: from alpha to omega.
\newblock In {\em SODA}, 2024.

\bibitem[YZ05]{yuster2005fast}
Raphael Yuster and Uri Zwick.
\newblock Fast sparse matrix multiplication.
\newblock {\em ACM Transactions On Algorithms}, 1(1):2--13, 2005.

\bibitem[Zwi02]{journals/jacm/Zwick02}
Uri Zwick.
\newblock All pairs shortest paths using bridging sets and rectangular matrix
  multiplication.
\newblock {\em Journal of the {ACM}}, 49(3):289--317, 2002.

\end{thebibliography}

\end{document}